%% file: min_sum_squar_main.tex
\documentclass[11pt]{article}
\usepackage{times}
\usepackage{amsmath}
\usepackage{amssymb}
\usepackage{subfigure}

\long\def\LongVersion#1\LongVersionEnd{#1}
\long\def\ShortVersion#1\ShortVersionEnd{}

\newcommand{\dgm}[0]{\textsf{DGM}}
\newcommand{\randdgm}[0]{\textsf{randomized DGM}}
\newcommand{\pboom}[0]{\textsf{parameterized boomerang}}
\newcommand{\pb}[0]{\textsf{PB}}

\newcommand{\qedd}{\nopagebreak \hfill $\Box$}

\usepackage{tikz}
\usetikzlibrary{arrows,decorations,decorations.shapes,backgrounds,shapes}

\newcommand{\Comment}[1]{}

\long\def\LongVersion#1\LongVersionEnd{#1}
\long\def\ShortVersion#1\ShortVersionEnd{}

\newenvironment{AvoidOverfullParagraph}[0]
{\sloppy\ignorespaces}
{\par\fussy\ignorespacesafterend}

\newtheorem{theorem}{Theorem}

\newtheorem{corollary}[theorem]{Corollary}
\newtheorem{example}[theorem]{Example}

\newtheorem{lemma}[theorem]{Lemma}
\newtheorem{definition}[theorem]{Definition}
\newtheorem{proposition}[theorem]{Proposition}
\newtheorem{claim}[theorem]{Claim}

\newenvironment{proof}{\noindent\bf{Proof.}\rm}{\hfill$\blacksquare$\bigskip}
\newenvironment{proofsketch}{\noindent\bf{Proof Sketch.}\rm}{\hfill$\blacksquare$}

\newcommand{\obsref}[1]{\mbox{Observation~\ref{#1}}}
\newcommand{\secref}[1]{\mbox{Section~\ref{#1}}}

\newcommand{\thref}[1]{\mbox{Theorem~\ref{#1}}}

\newcommand{\corref}[1]{\mbox{Corollary~\ref{#1}}}
\newcommand{\lemref}[1]{\mbox{Lemma~\ref{#1}}}

\newcommand{\propref}[1]{\mbox{Proposition~\ref{#1}}}
\newcommand{\figref}[1]{\mbox{Figure~\ref{#1}}}
\newcommand{\reals}{\mathbb{R}}
\newcommand{\bx}{\mathbf{x}}
\newcommand{\by}{\mathbf{y}}
\newcommand{\bz}{\mathbf{z}}
\newcommand{\bw}{\mathbf{w}}

\newcommand{\vecf}{\mathbf{f}}
\newcommand{\bxbar}{\mathbf{\bar{x}}}
\newcommand{\ybar}{\bar{y}}
\newcommand{\half}{\frac{1}{2}}

\newtheorem{mechanism}{Mechanism}

\begin{document}

\title{Strategyproof Location on a Network and the mini-Sum-of-Squares Objective}

\author{
Michal Feldman\thanks{Michal Feldman's work was partially supported by the Israel Science Foundation (grant number 1219/09), by the Leon Recanati
Fund of the Jerusalem School of Business Administration, the Google Inter-university center for Electronic Markets
and Auctions, and the People Programme (Marie Curie Actions) of the European Union's Seventh Framework
Programme (FP7/2007-2013) under REA grant agreement number 274919.
} \\
Hebrew University of Jerusalem and Harvard University\\
\texttt{mfeldman@huji.ac.il} \\
\and
Yoav Wilf \\
Hebrew University of Jerusalem\\
\texttt{yoavwilf@gmail.com} \\
}

\date{}
\maketitle

\begin{abstract}

We consider the problem of locating a public facility on a tree, where a set of $n$ strategic agents report their \emph{locations} and a mechanism determines, either deterministically or randomly, the location of the facility.
Game theoretic perspectives of the facility location problem have advanced in two main directions,
one focusing on the characterization of \emph{strategyproof} (SP) mechanisms and the other quantifying how well various objective functions can be approximated by SP mechanisms.
The present paper advances the theory in both directions.
First, we give a family of randomized SP mechanism for any tree network.
Quite miraculously, all of the deterministic and randomized SP mechanisms that have been previously proposed for various objective functions are special cases of our mechanism.
Thus, our mechanism unifies much of the existing literature on SP approximation for facility location problems.
Second, we use our mechanism to prove new bounds on the approximation of the minimum sum of squares (\emph{miniSOS}) objective in tree networks.
For randomized mechanisms, we propose a mechanism that gives 1.5-approximation for the miniSOS function on the line, and show that no other randomized SP mechanism can provide a better approximation.
For tree networks, we propose a mechanism that gives 1.83-approximation.
This result provides a separation between deterministic and randomized mechanisms, as we show that no deterministic mechanism can approximate the miniSOS function within a factor better than 2, even for the line.
Together, our study establishes a step toward the characterization of randomized SP facility location mechanisms on a tree, and provides a fundamental understanding of the miniSOS objective function.
\end{abstract}




\renewcommand{\thepage}{}
\clearpage
\pagenumbering{arabic}

\section{Introduction}
\label{sec:introduction}

In facility location problems, a social planner has to determine the location of a public facility that needs to serve a set of agents.
Once the facility is located, each agent incurs some cost that depends on the distance from her ideal location to the chosen location of the facility.
This class of problems is realized in many scenarios, including, for example, locating a server in telecommunication networks or locating a library or a fire station in a road network.
Facility location problems arise not only in physical settings like the ones just described, but also in more virtual settings where the agents' opinions or preferences can be represented as their locations, and a single outcome has to be chosen.
As an example, consider a set of students sitting in a classroom with an air conditioner, where every student has her most preferred temperature, and a single temperature has to be chosen.
In all of these examples, one ``location" has to be chosen, and every agent would like it to be as near as possible to her most preferred location.

What is the best method to settle the naturally conflicting preferences of the agents?
In other words, what is the best way to determine the location of the facility, given the agents' locations?
This question, and related ones, have been extensively studied in the literature, from various conceptual and algorithmic perspectives; see, e.g., Marsh and Schilling~\cite{MS94}, the book by Handler and Mirchandani~\cite{HM79}, and the body of literature concerning the performance of a \emph{Condorcet} point\footnote{A Condorcet point is a location that is preferred to any other location by more than half the agents.} \cite{Bandelt85,BL86,Labbe85}.
Starting with the work of Moulin~\cite{Moul80}, this problem was also studied from a game-theoretic perspective, where the agents are assumed to be strategic, i.e., report their locations in a way that will minimize their individual costs.
The game theoretic aspects of the facility location problem advanced in two main directions, as described below.

The first direction seeks to characterize \emph{strategyproof} (SP) mechanisms; i.e., mechanisms that induce truthful reporting as a dominant strategy.
This is a crucial attribute since in its absence, strategic agents may misreport their locations.
Moulin~\cite{Moul80} and later Schummer and Vohra~\cite{SV04} provided characterizations of \emph{deterministic} SP mechanisms on line, tree, and cycle networks.
While the work of Moulin~\cite{Moul80} concentrated on general single-peaked\footnote{With single-peaked preferences, every agent is associated with an ideal location, considered to be her \emph{peak}, and the closer the facility is to an agent's peak, the most preferred it is.} preferences~\cite{Black58}, Schummer and Vohra~\cite{SV04} considered the special case in which the cost incurred by an agent is the length of the shortest path from the facility to her location.
These results were later extended to various metric spaces (see, e.g., Border and Jordan \cite{BJ83}).

The second direction, advocated more recently by Procaccia and Tennenholtz\cite{PT09}, seeks to study the approximation ratio that can be obtained by an SP mechanism with respect to a given objective function.
This agenda, often termed ``approximate mechanism design without money", has led to extensive work on several domains, including facility location \cite{PT09,Alon09,Alon10,Lu09,Lu10,DT10}, machine learning~\cite{DFP10}, and matching~\cite{Ash10,Dg10}.
Unlike the traditional motivation for approximation, originating from computational hardness, here, approximation is used to achieve strategyproofness.

The approximation ratio of a mechanism is defined with respect to a given objective function, and the standard worst-case notion is being applied.
The two objective functions that have been prominently featured in the literature are the \emph{minisum} (i.e., minimizing the sum of agents' costs) and the \emph{minimax} (i.e., minimizing the maximum cost of any agent) functions.
While it is well known that the optimal location on a tree with respect to the minisum function can be obtained by an SP mechanism (in particular, by the {\em median}), it has been shown by Procaccia and Tennenholtz \cite{PT09} that no deterministic (respectively, randomized) SP mechanism on a line can achieve a better approximation than $2$ (resp., $1.5$) with respect to the minimax function.
Thus, the feasible approximation of SP mechanisms depends on the specified objective function.

In this work, we propose a third objective function --- minimizing the sum of squares of distances (hereafter {\em miniSOS}).
The miniSOS function is highly relevant in many economic settings, and is related to central notions in other disciplines, such as the {\em centroid} in geometry, or the {\em center of mass} in physics.
Of particular interest is the close relation to {\em regression learning}, where there is a one-to-one mapping between facility location on a line and regression learning (restricted to the class of constant functions).
Moreover, the miniSOS objective has been given a rigorous foundation by Holzman \cite{Holzman1990} using an \emph{axiomatic} approach.
In particular, Holzman formulated three axioms\footnote{
In particular, Holzman formulated three axioms, which are (roughly speaking): (i) unanimity, stating that if all the agents report the same location, then the reported location should be chosen, (ii) Lipschitz, which is a continuity requirement; and (iii) invariance, stating that an agent who moves to a location that is equidistant from the outcome from the same direction will not affect the chosen outcome.
}
for locating a facility on a line or a tree, regarded as sensible requirements, and showed that the unique objective function that satisfies the three axioms is the miniSOS objective.
Following the above motivation and the axiomatic foundation laid by Holzman, it is only natural to study the approximation that can be obtained by SP mechanisms with respect to the miniSOS function.

\subsection{Our contribution}

Our contribution spans both the characterization and the approximation agendas with respect to the line and tree networks.
The line topology is motivated by any one-dimensional decision that is made by aggregating agents' preferences.
The tree topology is particularly motivated by the telecommunication networks example, where a tree topology corresponds to any hierarchical network.
Notably, in computer networks, an agent can easily manipulate its perceived network location by generating a false IP address.

\paragraph{Characterization}

The characterization of SP mechanisms for facility location settings has been studied thus far mainly with respect to \emph{deterministic} mechanisms.
A {\em randomized} mechanism receives a location profile and returns a probability distribution over locations.
It is well known that randomization can be very powerful in social choice settings.
In cases where the cost incurred by an agent is her distance from the facility (as in the model of Schummer and Vohra~\cite{SV04}), it seems natural to define an agent's cost as her {\em expected cost} with respect to the given probability distribution.
This is the approach taken by Procaccia and Tennenholtz \cite{PT09}, and many studies thereafter \cite{Alon09,Alon10,Lu10}\footnote{These works, however, focused on approximating various objective functions by SP mechanisms and not on characterizing SP mechanisms.}.
Quite surprisingly, while there has been a surge of research on randomized SP mechanisms for facility location problems, there was very limited effort on characterizing randomized SP mechanisms\footnote{An exception is the work by Ehlers {\em et al.} \cite{EPS02}, which characterizes randomized SP mechanisms, but is different from our work in two respects. First, the characterization of Ehlers {\em et al.} applies only to a line network, while we are interested in the more general case of a tree. Second, they use a different notion of preferences over probability distribution (in particular, preference over probability distributions is defined in terms of first-order stochastic dominance), and, as a result, uses a different notion of strategyproofness. This difference has significance effects, for example, on lower bound results for approximation.}.

We make a first step in providing a characterization for randomized SP mechanisms on a {\em tree} network.
In particular, we design a family of randomized mechanisms, termed \pboom{} (\pb{}), which is SP for any tree network.
Quite miraculously, all of the mechanisms that have been devised recently within the facility location domain (and shown to give tight bounds with respect to various objective functions) can be formulated as special cases of mechanism \pb{}.
The generality and strength of mechanism \pb{} is further illustrated through the analysis of the miniSOS objective, where special cases of mechanism \pb{} are shown to achieve good approximation results on line and tree networks with respect to the miniSOS objective.
Thus, the characterization of randomized SP mechanisms is not only interesting in itself but also equips one with a large family of SP mechanisms, and can be served as a useful tool for studying the approximation ratios with respect to various objectives.
Mechanism \pb{} is presented in \secref{sec:tree-general-mechanism}.


\paragraph{Approximation}
We provide approximation results with respect to the miniSOS objective, for line and tree networks, and for deterministic and randomized mechanisms.
Notably, all the mechanisms that are devised in this paper are special cases of mechanism \pb{}.

{\em Line networks} (\secref{sec:SP-mechanisms-line}).
We show that the {\em median} gives a $2$-approximation deterministic SP mechanism, and that no deterministic SP mechanism can achieve a better approximation ratio.
As in other studies~\cite{PT09,Alon10,Lu10,Ash10}, randomized mechanisms are shown to provide better bounds.
In particular, we present a randomized SP mechanism that provides a $1.5$-approximation; the mechanism chooses the average location with probability $\frac{1}{2}$ and a random dictator with probability $\frac{1}{2}$.
In addition, we show that no randomized SP mechanism can achieve a better bound.
The proof technique used to construct the lower bound involves subtle analysis and is perhaps one of the main technical contributions of the paper.
Interestingly, while the minimax and the miniSOS functions induce different optimal solutions, and different optimal SP mechanisms, they admit the exact same approximation bounds with respect to both deterministic and randomized mechanisms.

{\em Tree networks} (\secref{sec:tree}).
First, we show that the median gives a $2$-approximation with respect to the miniSOS objective.
This result is tight with respect to deterministic mechanisms, following the lower bound established on the line.
Our main result in the approximation regime is the construction of a randomized SP mechanism that gives a $1.83$-approximation for any tree network.
This result establishes a separation between deterministic and randomized mechanisms, as no deterministic mechanism can provide a better approximation than 2.

\subsection{Open Problems}
\label{sec:open-problems}
Our study leaves many questions for future research.
The first natural challenge is to provide a full characterization for SP randomized mechanisms on a tree.
For the miniSOS objective function, closing the approximation gap for the randomized mechanisms on a tree remains open.
In addition, it would be interesting to extend the approximation results for the miniSOS function to other networks topologies, such as a cycle and general networks, similar to the studies by Alon {\em et al.}~\cite{Alon09,Alon10} with respect to the minimax and the minisum functions.
An additional direction is to consider a different individual cost function.
For example, in applications in which agents are more sensitive to distances within the range of high distances, a convex cost function seems plausible (a possible example might be the speed of an Internet connection).
Finally, the three different social functions that have been studied thus far can be considered as special cases of the $\ell$-norm distance, with minisum, miniSOS, and minimax corresponding to the $1$-norm, $2$-norm, and $\infty$-norm, respectively.
It is apparent that while for the minisum function, the optimal location can be obtained in an SP mechanism, this is not feasible for either $2$- or $\infty$-norms, and the same approximation bounds apply in both cases.
Generalizing this result to any $\ell$-norm is an additional stimulating direction.

\section{Model and Preliminaries}
\label{sec:model}

We use the model of Schummer and Vohra~\cite{SV04}, where the network is represented by a graph $G$, formalized as follows. The graph is a closed, connected subset of Euclidean space $G \subseteq \reals^k$.
The graph is composed of a finite number of closed curves of finite length, known as the edges\footnote{
Note that while this model is expanding upon the notion of an interval, it is not analyzing
full-dimensional, convex subsets of Euclidean space. Rather, travel is restricted
to a road network, where convex combinations of locations are typically not feasible.}.
The extremities of the curves are known as vertices (or nodes).
An important class of graphs, which is the focus of this paper, is {\em tree} graphs --- graphs that contain no cycles.

The path between two points $a,b \in G$ is denoted by $path_G(a,b)$.
The distance between two points $a,b \in G$, denoted $d_G(a,b)$, is the length of the (unique) path between $a$ and $b$.
We extend the definition of distance between points to distance between a point and a path as follows.
Given a point $c \in G$ and a path $path_G(a,b)$, the distance between $c$ and $path(a,b)$, denoted $d_G(path_G(a,b),c)$, is the shortest distance between $c$ and any point on $path(a,b)$; i.e., $d_G(path_G(a,b),c) = min_{l \in path_G(a,b)} d_G(l,c)$. When clear in context, we omit the subscript $G$.

Let $N=\{1,\ldots,n\}$ be a set of agents.
We sometime use $[n]$ to denote the set of agents $N$.
Each agent $i\in N$ has an (ideal) location $x_i\in G$ (agents can be located anywhere on $G$).
The collection $\bx=( x_1,\ldots,x_n ) \in G^n$ is referred to as the {\em location profile}.

A \emph{deterministic mechanism} is a function $f:G^n\rightarrow G$ that maps the agents' reported locations to the location of a {\em facility} (which can be located anywhere on $G$).
If the facility is located at $y \in G$, the cost of agent $i$ is the distance between $x_i$ and $y$; i.e., $cost(y,x_i) = d(y,x_i)$.

A \emph{randomized mechanism} is a function $f: G^n \rightarrow \Delta(G)$, which maps location profiles to probability distributions over $G$ (which randomly designate the facility location).
Let $P \in \Delta(G)$ be a probability distribution over $G$.
If $f(x)=P$, then the cost of agent $i$ is the expected distance of the facility location from $x_i$;
i.e., $cost(P,x_i) = E_{y \sim P}[cost(y,x_i)]$.
When clear in the context, we write $y \sim f(\bx)$ for ease of presentation.

A mechanism is called \emph{strategyproof} (SP), or \emph{truthful}, if no agent can benefit from misreporting her location, regardless of the reports of the other agents. Formally, in our scenario, this means that for all $\bx \in G^n$, for all $i \in N$, and for all $x'_i \in G$, it holds that $cost(f(\bx),x_i) \leq cost(f(x'_i,x_{-i}),x_i)$, where $x_{-i}= ( x_1, \ldots, x_{i-1},x_{i+1}, \ldots, x_n ) $ is the profile of all locations, excluding agent $i$'s location.

The quality of a facility location is usually evaluated with respect to some target social function.
Given a location profile $\bx=(x_1,\ldots,x_n)$ and a facility location $y$, the social cost of $y$ with respect to $\bx$ is given by a function $sc(y,\bx)$.
The social cost of a distribution $P$ with respect to $\bx$ is $sc(P,\bx)=E_{y \sim P}[sc(y,\bx)]$.

Given a social cost function, location $y \in G$ is said to be {\em optimal} with respect to a profile $\bx$ if
$sc(y,\bx) = min_{y' \in G}sc(y',\bx)$.
An optimal location is denoted by $Opt(G,\bx)$.
When clear in the context, we simply write $Opt$.
In addition, we often abuse notation and use $Opt$ to refer to the social cost of an optimal location.

A mechanism $f$ is said to provide $\alpha$-approximation with respect to a social cost function $sc$ if for every graph $G$ and every location profile $\bx$, $sc(f(\bx),\bx)/sc(Opt,\bx) \leq \alpha$;
that is, the mechanism always returns a solution that is an $\alpha$ factor of the optimal solution.

In this paper we are interested in optimizing the sum of squared distances (SOS) function; that is,
$sc(y,\bx)=\sum_{i \in N}d(y,x_i)^2$.
This objective function is extremely important from both normative and positive perspectives, as discussed in the introduction.

Given a profile $\bx$, the median of $\bx$ in a tree $G$, denoted by $\mu(G,\bx)$, is defined as follows. We start from an arbitrary node (induced by $G$) as a root. Then, as long as the current location has a subtree that contains more than half of the agents, we smoothly move down this subtree. Finally, when we reach a point where it is not possible to move closer to more than half the agents by continuing downwards, we stop and return the current location.

We continue with several graph theoretic definitions and lemmas. At this point, it is necessary to emphasis the difference between a \emph{location profile}, which was defined earlier and is tightly coupled with a set of agents, and a \emph{location vector}, which is a set of locations in the graph.

\begin{definition}
Given a tree $G$ and a point $x \in G$, let $T(G,x)$ be the set of subtrees defined as follows.
If $x$ is a tree node (with degree $d_x$), then $T(G,x)=\{T_1, \ldots, T_{|d_x|}\}$, where $T_i$ is the subtree of descendant $i$ rooted at $x$. If $x$ is not a node (i.e., it is a point on an edge), then $T(G,x)=\{T_1,T_2\}$, where $T_1$ and $T_2$ are the respective left and right subtrees rooted at $x$.
\end{definition}

\begin{definition}
\label{def:weighted-opt}
Let $G$ be a tree, $\by \in G^m$ be a location vector, and $\bw$ be a probability vector of size $m$.
The weighted average location with respect to $G,\by$ and $\bw$, denoted $wAvg(G,\by,\bw)$, is a point in $G$ which minimizes the weighted sum of squared distances from the locations in $\by$; i.e., $wAvg \in argmin_{l \in G}\sum_{j \in [m]}w_j d(l,y_j)^2$.
\end{definition}

The following lemmas will be required in the sequel. Their proofs are deferred to Appendix~\ref{sec:model-appendix}.

\begin{lemma}
\label{lem:weighted-opt-derivative}
Let $\by \in G^m$ be a location vector, and $\bw$ a probability vector of size $m$.
It holds that $a = wAvg(G,\by,\bw)$ if and only if for every $T_j \in T(G,a)$,
$$
\sum_{i \in [m]: y_i \in T_j}  w_i d(y_i,a) \leq \sum_{i \in [m]: y_i \notin T_j}  w_i d(y_i,a).
$$
\end{lemma}

\begin{corollary}
\label{cor:wAvg-is-unique}
Let $\by \in G^m$ be a location vector, and $\bw$ a probability vector of size $m$.
The weighted average location with respect to $G,\by$ and $\bw$ is unique.
\end{corollary}

\begin{lemma}
\label{lem:movement-of-weighted-opt}
Let $\by,\by' \in G^m$ be location vectors, and $\bw$ be a probability vector of size $m$.
For every $i \in [m]$, let $\delta_i = d(y_i,y'_i)$ .
Let $a = wAvg(G,\by,\bw)$ and $a' = wAvg(G,\by',\bw)$.
It holds that $d(a,a') \leq \sum_{i \in [m]}  w_i \delta_i$.
\end{lemma}

\begin{lemma}
\label{lem:differnce-of-costs-two-on-line}
Let $\bx \in G^n$ be a location profile. Let $a,b \in G$ be two locations, and let $T_b \in T(G,a)$ and $T_a \in T(G,b)$ such that they contain $path(a,b)$, and assume that all of the agents are either in $G \backslash T_a$, in $G \backslash T_b$, or on $path(a,b)$. Additionally, assume that $Opt$ is located in $G \backslash T_a$. Then, $sc(a,\bx) - sc(b,\bx) = -|N|d(a,b)^2 - 2d(a,b) \left( \sum_{i \in G \backslash T_b}d(x_i,a) - \sum_{i \in T_b}d(x_i,a) \right)$
\end{lemma}

\begin{corollary}
\label{cor:opt-flattening}
Let $\bx \in G^n$ be a location profile, and let $a,b \in G$ as described in \lemref{lem:differnce-of-costs-two-on-line}. Let $\bx'$ be constructed as follows. Let $x_j \in G \backslash T_a$ be the most far agent from $b$ in $G \backslash T_a$, and for each $i$ such that $x_i \in G \backslash T_a$, locate $x'_i$ on $path(b,x_j)$ such that $d(b,x'_i)=d(b,x_i)$. Then, $sc(a,\bx) - sc(b,\bx) = -|N|d(a,b)^2 + 2|N|d(a,b) d(a,Opt')$
\end{corollary}

\section{Randomized SP Mechanisms on a Tree}
\label{sec:tree-general-mechanism}


In this section we introduce a family of randomized SP mechanism for locating a facility on a tree.
Unless otherwise stated, the graph $G$ in this section is assumed to be a tree.


\vspace{0.1in}

The following notion of a {\em boomerang} mechanism is a key concept in our construction.

\begin{definition}
\label{def:boomerang-mechanism}
A deterministic mechanism $f$ is said to be a {\em boomerang} mechanism if for every location profile $\bx$, agent $i$, and point $x_i'$,
$cost(f(\bx'),x_i) - cost(f(\bx),x_i) = d(f(\bx'),f(\bx))$, where $\bx'=(x_i',x_{-i})$.
\end{definition}
That is, a boomerang mechanism is one in which a deviating agent fully absorbs the effect of her deviation on the facility location.
Clearly, every boomerang mechanism is SP.


Several examples of boomerang mechanisms follow:
(The proof is left to the reader.) (i) {\em dictatorship}; i.e., where there exists $i \in N$ such that for every $\bx$, $f(\bx)=x_i$. (ii) {\em median} (on a tree). (iii) {\em $k$'th-location} (on a line); also known as {\em generalized median} \cite{Moul80}.



We are now ready to introduce the family of randomized SP mechanisms for tree networks.
This family is presented as a parameterized mechanism, called ``parameterized boomerang".

\vspace{0.1in}

\noindent{\bf Mechanism \pboom{} \pb{}:}
Let $\vecf = ( f_1, \ldots, f_m )$ be a collection of boomerang mechanisms.
For every $i \in [m]$, let $y_i = f_i(\bx)$, and let $\by=( y_1, \ldots, y_m )$.
Let $\bw$ be a probability distribution supported on $m$ elements, and let $a = wAvg(G,\by,\bw)$.
The facility location is chosen according to the following probability distribution:
\begin{itemize}
\item for every $i \in [m]$, choose $f_i(\bx)$ with probability $\half w_i$.
\item choose $a$ with probability $\half$.
\end{itemize}
We refer to the two components of the probability distribution as the \emph{boomerang} component and \emph{average} component, respectively.
Note that every boomerang mechanism is a special case of \pb{}, with $m=1$.

The following theorem, whose proof is deferred to Appendix~\ref{sec:tree-general-mechanism-appendix}, establishes the strategyproofness of Mechanism \pb{}. It is followed by an immediate corollary.


\begin{theorem}
\label{thm:tree-snc-is-truthful}
Mechanism \pb{} is SP.
\end{theorem}

\begin{corollary}
\label{cor:snc-dist-is-truthful}
Every fixed probability distribution over \pb{} mechanisms is SP.
\end{corollary}


In recent years, various SP mechanisms have been proposed in the literature for the facility location problem on the line with the objective of approximating different social objectives, such as the minisum and the minimax functions.
The following proposition shows that all of the mechanisms that have been proposed in this context are special cases of Mechanism \pb{} (or a probability distribution over \pb{} mechanisms).
The proof is deferred to Appendix~\ref{sec:tree-general-mechanism-appendix}.

\begin{proposition}
\label{prop:exmaples}
The following mechanisms on the line are special cases of Mechanism \pb{} (or a probability distribution over \pb{} mechanisms).
\begin{enumerate}
\vspace{-3mm}
\item {\em $k$-location}. Examples of this mechanism are the median mechanism, which is known to minimize the sum of distances, and the leftmost agent mechanism, which provides a $2$-approximation for the minimax objective \cite{PT09} (which is tight with respect to deterministic mechanisms).
\vspace{-3mm}
\item {\em left-right-middle} (LRM) \cite{PT09}. LRM chooses the leftmost agent with probability $\frac{1}{4}$, the rightmost agent with probability $\frac{1}{4}$, and their middle point with probability $\frac{1}{2}$. It provides a (tight) $1.5$-approximation for the minimax objective.
\vspace{-3mm}
\item {\em random dictator} (RD). RD chooses every agent with probability $\frac{1}{n}$. It provides a $2-\frac{2}{n}$-approximation for the minisum objective \cite{Alon09}.
    We will later establish that RD gives a $2$-approximation for the miniSOS objective (see Theorem~\ref{thm:random-2-approx}).
\vspace{-3mm}
\end{enumerate}
\end{proposition}

While it is evident from the last proposition that Mechanism \pb{} is very powerful, it imposes a sufficient condition for strategyproofness, but not a necessary one. This is established by Example~\ref{exm:counter-necessary} that could be found in Appendix~\ref{sec:tree-general-mechanism-appendix}.

\section{SP Mechanisms on a Line}
\label{sec:SP-mechanisms-line}
In this section, we study how well SP mechanisms can approximate the miniSOS objective --- minimizing the sum of squared distanced --- on a line.
In the deterministic case, we present a mechanism that provides $2$-approximation, and show that no SP deterministic mechanism can achieve a better ratio.
In the randomized case, we construct a mechanism that provides $1.5$-approximation, and show that this result is tight.
All missing proofs in this section are deferred to Appendix~\ref{sec:SP-mechanisms-line-appendix}.

In this section, the graph is essentially the {\em real} line, $\reals$.
It is easy to verify that an optimal location in this case is simply the average.

\begin{claim}
\label{lem:optimal-point}
Given a location profile $\bx$, the optimal facility location with respect to the miniSOS objective is the average location; i.e., $Opt = argmin_y sc(y,\bx) = \frac{\sum_{i \in N} x_i}{n}$.
\end{claim}

The following lemma proves extremely useful in establishing the lower bounds throughout this section.
In particular, it helps us relate joint deviations (i.e., coordinated deviations by a subset of agents) to unilateral deviations (i.e., deviations by a single agent).
\begin{lemma}
\label{lem:immigrants}
Let $a,b,c \in \reals$ be three locations such that $a \leq b \leq c$, with at least one strict inequality, and for every $m \in [n]$, let $\bx^0$ (respectively, $\bx^m$) be a location profile in which $n-m$ agents are located at $a$, and $m$ agents are located at $c$ (resp., $b$).
Let $f$ be a randomized mechanism.
If $f$ is an SP mechanism, then $E[|c-y^0|] \leq E[|c-y^m|]$ and $E[|b-y^m|] \leq E[|b-y^0|]$, where $y^0 \sim f(\bx^0)$ and $y^m \sim f(\bx^m)$.
\end{lemma}


\subsection{Deterministic mechanisms}

\begin{theorem}
\label{thm:median-two-approx}
Given a location profile $\bx$, the mechanism that chooses the median location in $\bx$ is an SP $2$-approximation mechanism for the miniSOS objective.
\end{theorem}

Notably, this mechanism is a special case of mechanism \pb{}.


The following theorem shows that factor $2$ is tight with respect to deterministic SP mechanisms.
\begin{theorem}
\label{thm:deterministic-bound}
Any deterministic truthful mechanism has an approximation ratio of at least $2$ for the miniSOS objective.
\end{theorem}


\subsection{Randomized mechanisms}
\label{subsec:SP-mechanisms-line-random}

A natural candidate of a randomized mechanism  to be considered in our context is the \emph{random dictator} (RD) mechanism, which chooses each agent's location with probability $\frac{1}{n}$. This mechanism is SP and is known to provide a $\left( 2-\frac{2}{n} \right)$-approximation with respect to the minisum objective function (See \cite{Alon09} and \cite{Alon10}).
The following theorem shows that the RD mechanism provides a $2$-approximation for the miniSOS objective.
More precisely, for \emph{every} location profile, the RD mechanism yields an SOS cost that is exactly twice the cost of the optimal location.
This is established in the following theorem.

\begin{theorem}
\label{thm:random-2-approx}
For every location profile, the RD mechanism yields an SOS cost that is exactly twice the optimal SOS cost.
\end{theorem}

As shown in \propref{prop:exmaples}, the RD mechanism is a probability distribution over \pb{} mechanisms.

Apparently, the RD mechanism does not perform better the deterministic median mechanism.
Yet, this mechanism turns out to be useful when integrated within a more sophisticated mechanism, as shown below.

\begin{mechanism}
\label{mech:half_opt_half_RD}
Given $\bx\in R^n$, choose the average point with probability $\frac{1}{2}$, and apply the RD mechanism with probability $\frac{1}{2}$ (i.e., for every $i \in N$, $x_i$ is chosen with probability $\frac{1}{2n}$).
\end{mechanism}


\begin{theorem}
\label{thm:random-1.5-approx}
Mechanism~\ref{mech:half_opt_half_RD} is an SP $1.5$-approximation for the SOS objective.
\end{theorem}

Notably, while approximation in its usual sense looks at the worst-case ratio between the expected cost of the mechanism's solution and the cost of the optimal solution, in this case the $1.5$-approximation applies not only in the worst-case notion; rather, this is the exact approximation achieved for every location profile.

Surprisingly, Mechanism~\ref{mech:half_opt_half_RD} provides the best possible approximation; that is, no SP mechanism, randomized or not, can achieve a better approximation ratio than $1.5$.
This bound is established in the next theorem. Due to lack of space, we give an overview of the proof, and defer the full proof to Appendix~\ref{sec:SP-mechanisms-line-appendix}.
%

\begin{theorem}
\label{thm:random-bound}
Any randomized SP mechanism has an approximation ratio of at least $1.5$ for the miniSOS objective.
\end{theorem}

\begin{proofsketch}
Assume on the contrary that there exist an SP mechanism $f(\bx)$ and $\epsilon>0$ such that $f(\bx)$ yields an approximation ratio of $1.5-\epsilon$, and let $y$ be a random variable that is distributed according to $f(x)$. Starting at the original profile, we construct a series of profiles, on which we repeatedly apply \lemref{lem:immigrants}.
Using this technique, we prove that there exist some $a \in \reals$ and a location profile $\bx$ in which $\frac{n}{2}$ agents are located at $a$ and $\frac{n}{2}$ agents are located at $a+4$, such that $E[|y-(a+1)| + |y-(a+3)|]>3-2\epsilon$.
We then use the last inequality to show that $E[sc(y,\bx)] > 6n - 4n\epsilon$.
By observing that the optimal facility location for this profile is at $a+2$ (which yields an SOS cost of $4n$),
we get $\frac{E[sc(y,\bx)]}{Opt} > \frac{6n - 4n\epsilon}{4n} = 1.5 - \epsilon$, and a contradiction is reached.
\end{proofsketch}

\section{SP Mechanisms on a Tree}
\label{sec:tree}
In this section, we study the miniSOS objective with respect to locating the facility on a tree.
In the deterministic case, we show that the median of a tree provides a $2$-approximation for the miniSOS objective, and show that no SP deterministic mechanism can achieve a better ratio with respect to miniSOS.
In the randomized case, we construct an instance of Mechanism \pb{}, which obtains a $1.83$-approximation.

\begin{figure}
\label{fig:tree-median}
\begin{center}
\input{draw-tree-median-before}
\input{draw-tree-median-after}
\end{center}
\caption{An illustration of the iterative process described in the proof of \thref{thm:tree-det-mechanism}. $G^{j-1}$ is illustrated on the left, and $G^j$ is illustrated on the right, where the dashed line represents the new edge.}
\label{fig:median-process}
\end{figure}
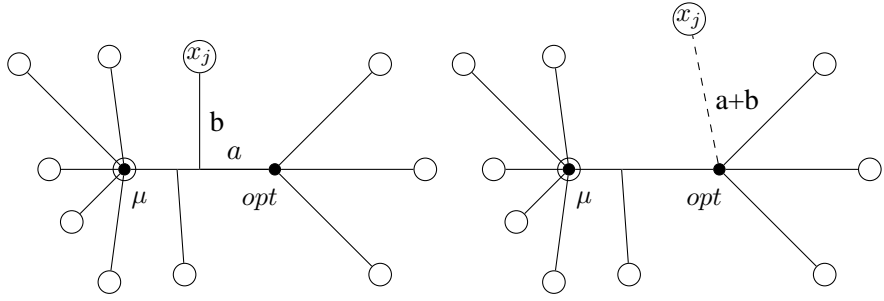


\subsection{Deterministic mechanisms}

It is well known that the mechanism that chooses a median of a tree is SP (see, e.g., \cite{Alon09}).
The following theorem establishes that a median also gives a $2$-approximation for the miniSOS objective, which is tight, according to \thref{thm:deterministic-bound}.
The theorem is followed by a sketch of the proof.
The full proof is deferred to Appendix~\ref{sec:tree-appendix}.

\begin{theorem}
\label{thm:tree-det-mechanism}
The median of a tree is an SP $2$-approximation mechanism for the miniSOS objective.
\end{theorem}


\begin{proofsketch}
Given an instance $(G,\bx)$, we iteratively transform it, in a way that can only make the approximation ratio obtained by the median worse, and eventually prove the desired approximation ratio on the final instance.
Let $\mu$ and $Opt$ denote the respective median and optimal location in the original profile.
The iterative process proceeds as follows.
As long as there exists an agent $j$ such that $x_j$'s subtree is rooted at the open interval $path(\mu,Opt)$, pick such an agent $j$, create a new edge of length $d(x_j,Opt)$, rooted at $Opt$, and locate $x_j$ at its tip
(see illustration in \figref{fig:median-process}).
It can be proved that the median and the optimal location did not change as a result of this transformation, and that the optimal cost did not change either.
The cost of the median, however, can only increase.
Therefore, the approximation ratio can only get worse by each transformation.
Upon termination of this process, we prove the desired approximation ratio on the final instance by reducing it to the deterministic scenario on a line and applying \thref{thm:median-two-approx}.
\end{proofsketch}


\subsection{Randomized mechanisms}

In this section we present an instance of mechanism \pb{}, which obtains $1.83$-approximation for trees, with respect to the miniSOS objective.
Our randomized mechanism uses the following deterministic mechanism as a building block. It is followed by a proposition, whose proof is deferred to Appendix~\ref{sec:tree-appendix}.

{\bf Mechanism \verb"dictatorial-generalized-median"} (\dgm{}):
Mechanism \dgm{} receives as parameters an index $i \in [n]$, and a fraction $q \in (1/2,1]$.
The facility location is chosen deterministically, with respect to $i$ and $q$, as follows.
Fix the point $x_i$ as the root of the tree, and denote the current location $a$.
Then, as long as there exists a subtree in $T(G,a)$ that contains at least fraction $q$ of the agents (this is well defined, since $q>1/2$), smoothly move down this subtree.
Finally, when we reach a point where it is not possible to move closer to at least fraction $q$ of the agents by continuing downwards, we stop and return the current location.

\begin{proposition}
\label{prop:PC-is-boomerang}
Mechanism \dgm{} is a boomerang mechanism.
\end{proposition}

Note that mechanism \dgm{} is not a boomerang mechanism for $q \leq 1/2$, as a beneficial misreport can exist, depending on the tie-breaking rule.

With this we are ready to introduce our randomized mechanism.

\vspace{0.1in}

{\bf Mechanism \verb"randomized DGM"}:
Mechanism \randdgm{} receives as a parameter a fraction $q \in (1/2, 2/3]$, and applies mechanism
\pb{} with the following parameters: $m=n$, $w_i=1/n$ for every $i \in [n]$, and for every $i \in [n]$, $f_i(\bx)$ is mechanism \dgm{} with parameters $i$ and $q$.

Interestingly, the following property holds. Its proof is deferred to the appendix.

\begin{lemma}
\label{lem:projected-path}
For every tree $G$, the points $f_1(\bx), \ldots, f_n(\bx)$, calculated by Mechanism \randdgm{}, are located on a single path.
\end{lemma}

In the remainder of this section, we shall use the notation suggested in \lemref{lem:projected-path}, i.e., for every $i \in [n]$, let $y_i = f_i(\bx)$. This is followed by another notation - $avg(\by)$, which represents the average location on $path(y_1,y_n)$, given the locations in $\by$.


The main result of this section establishes that Mechanism \randdgm{}, with $q=2/3$, obtains an approximation ratio of $1.83$ for trees.

%

\begin{theorem}
\label{thm:rndom-tree-1.82}
Let $G$ and $\bx$ be a tree and a location profile, respectively. Mechanism \randdgm{} with $q=2/3$ obtains an approximation ratio of $1.83$ with respect to the miniSOS objective.
\end{theorem}

Before presenting the proof sketch of the theorem, we observe that it can be assumed w.l.o.g that $y_1 \neq y_n$, as the following lemma establishes.

\begin{lemma}
\label{lem:y-is-one-location}
If $y_1=y_2=\ldots=y_n$, then the approximation ratio obtained by Mechanism \randdgm{} is at most $1.5$.
\end{lemma}

We are now ready to present an overview of the proof of Theorem~\ref{thm:rndom-tree-1.82}.

\begin{proofsketch}
The proof of this theorem requires several lemmata, which correspond to various transformations that are needed in order to prove the desired approximation ratio.
We shall describe the proof schematically along with a graphical illustration, but defer the formal statements of the lemmata and their proofs to Appendix~\ref{sec:tree-appendix}.
Following \lemref{lem:y-is-one-location}, we assume that $y_1 \neq y_n$.
Additionally, we assume, without loss of generality, that $d(y_1,y_n)=1$, and scale the graph accordingly.
We let $\bar{Opt}$ denote the closest location on $path(y_1,y_n)$ to $Opt$, i.e., $\bar{Opt} = min_{l \in path(y_1,y_n)} d(Opt,l)$. For ease of presentation, we refer to Mechanism \randdgm{} with $q=2/3$ as ``the mechanism''.

The proof is established as follows.
Given a graph $G$ and a location profile $\bx$, we run them through various transformations, and show that the approximation ratio obtained by the mechanism could only get worse in every transformation.
Eventually, we prove that the mechanism provides a $1.83$-approximation ratio on the final graph and location profile, which implies an upper bound of $1.83$ on the original instance.

More specifically, given an instance $(G,\bx)$, we proceed as follows:
\vspace{-1mm}
\begin{itemize}
\vspace{-3mm}
\item We transform the graph repeatedly according to \lemref{lem:flattening}, which results in a graph in which all of the agents that are not located on $path(y_1,y_n)$ are rooted at either $y_1$, $y_n$, or $\bar{Opt}$ (see Figure 2).
\vspace{-3mm}
\item \begin{AvoidOverfullParagraph}Assuming that $avg(\by)$ is on $path(y_1,\bar{Opt})$ (the proof follows analogously if $avg(\by)$ is on $path(\bar{Opt},y_n)$), we then transform the graph according to \lemref{lem:left-edge-lemma} and \lemref{lem:right-edge-lemma} such that all the agents rooted at $y_1$ (i.e., all agents $i$ such that $y_i=y_1$) are located at $y_1$, and all the agents rooted at $y_n$ (i.e., all agents $i$ such that $y_i=y_n$) are located on their own new edges, rooted at $y_n$, and equally distanced from it (see Figures 3 and 4).\end{AvoidOverfullParagraph}
\end{itemize}
\vspace{-3mm}
At this point, the transformed graph can be shown to have the following properties: At least $|N|/3$ agents are located at $y_1$; at least $|N|/3$ agents are located on new edges, rooted at $y_n$ and equally distanced from it; some agents are scattered along $path(y_1,y_n)$; and the rest are located on edges that are rooted at the same location on $path(y_1,y_n)$. Denote this location by $p$. The location $p$ is essentially the original location of $\bar{Opt}$ (note that after transforming the graph using \lemref{lem:left-edge-lemma}, $\bar{Opt}$ might be relocated).
We now distinguish between several cases, depending on $p$'s location:
\vspace{-1mm}
\begin{itemize}
\item If $p$ is located at $y_n$, then $Opt$ could be located either at $y_n$ as well, or in a subtree rooted at $p$.
 In the former case, a $1 \frac{1}{2}$-approximation is obtained, by \lemref{lem:pushing-M}.
 In the latter case, we transform the graph using \lemref{lem:M-edge-lemma} (see Figure 5), and by \lemref{lem:ratio-1.82} we achieve an approximation ratio of $1.83$.
\vspace{-3mm}
\item If $p$ is located on the open interval $path(y_1,y_n)$, we transform the graph using \corref{cor:M-edge-open-lemma} (see Figure 5), and again, by \lemref{lem:ratio-1.82}, we achieve an approximation ratio of $1.83$.
\end{itemize}
\vspace{-3mm}
By observing that $p$ cannot be located at $y_1$, we conclude that these two cases exhaust all the possibilities for $p$'s location.
In conclusion, since all the transformations are shown to only worsen the approximation ratio, the $1.83$-approximation that is obtained for the final instance, imposes the same upper bound on the original one, and the assertion of the theorem follows.


\end{proofsketch}

\begin{figure}[!ht]
\begin{center}$
\begin{array}{lr}
\input{draw-original-state} &
\input{draw-lemma1}
\end{array}$
\end{center}
\label{lemma1-figure}
\caption{\lemref{lem:flattening}'s transformation assures that all the agents that are not located on $path(y_1,y_n)$ are rooted at either $y_1$, $y_n$, or $\bar{Opt}$.}
\begin{center}$
\begin{array}{cc}
\input{draw-lemma1-arrow} &
\input{draw-lemma2}
\end{array}$
\end{center}
\caption{\lemref{lem:left-edge-lemma}'s transformation assures that all agents that are rooted at $y_1$ are in fact located at $y_1$.}
\end{figure}

\begin{figure}[!ht]
\begin{center}$
\begin{array}{cc}
\input{draw-lemma2-arrow} &
\input{draw-lemma3}
\end{array}$
\end{center}
\caption{\lemref{lem:right-edge-lemma} averages the distance of the agents that are rooted at $y_n$.}
\begin{center}$
\begin{array}{cc}
\input{draw-lemma3-arrow} &
\input{draw-lemma4}
\end{array}$
\end{center}
\caption{\lemref{lem:M-edge-lemma} and \corref{cor:M-edge-open-lemma} average the distance of the agents that are rooted at $\bar{Opt}$, while locating them on the same subtree.}
\label{fig:flattening-process}
\end{figure}

\bibliographystyle{plain}
\bibliography{bibfile}


\clearpage
\pagenumbering{roman}
\appendix

\begin{figure}[t]
\begin{center}
\textbf{\large{APPENDIX}}
\end{center}
\end{figure}

\parindent 0mm


\section{Missing proofs in \secref{sec:model}}
\label{sec:model-appendix}

\begin{proof}
of \lemref{lem:weighted-opt-derivative}:
Suppose by way of contradiction that there exists $T_j \in T(G,a)$, such that
\begin{equation}
\sum_{i \in [m]: y_i \in T_j} w_i d(y_i,a)> \sum_{i \in [m]: y_i \notin T_j} w_i d(y_i,a).
\label{eq:tree-10}
\end{equation}

Let $b$ be the distance between $a$ and the closest vertex or location in $\by$ in the subtree $T_j$ (such that $d(a,b)>0$).
Let $h(y)$ be a function $h:[0,b) \rightarrow \reals$ that yields the weighted sum of squared distances for the location that is distanced $y$ from $a$ on the subtree $T_j$.
It holds that
$$
h'(0^{+})  = \sum_{i \in [m]: y_i \notin T_j} w_i 2d(y_i,a) - \sum_{i \in [m]: y_i \in T_j} w_i 2d(y_i,a) < 0,
$$
where the last inequality follows by Equation~\eqref{eq:tree-10}. It follows that there exists a location for which the weighted SOS distances is lower than in $a$; hence we reach a contradiction.

{\em Sufficiency:}
Assume that for every $T_j \in T(G,a)$, it holds that
$$
\sum_{i \in [m]: y_i \in T_j}  w_i d(y_i,a) \leq \sum_{i \in [m]: y_i \notin T_j}  w_i d(y_i,a),
$$
and suppose by way of contradiction that $a$ is not the weighted average location.
Let $b$ be a weighted average location.
Let $T_{b} \in T(G,a)$ be the subtree such that $b \in T_{b}$.
Let $S_a \in T(G,b)$ be the subtree such that $a \in S_a$.

Since $T(G,b) \backslash S_a \subset T_{b}$ , it holds that
\begin{equation}
\sum_{i \in [m]: y_i \notin S_a} w_i d(y_i,b) < \sum_{i \in [m]: y_i \notin S_a} w_i d(y_i,a) \leq \sum_{i \in [m]: y_i \in T_{b}} w_i d(y_i,a).
\label{eq:tree-70}
\end{equation}
In a similar manner, since $T(G,a) \backslash T_{b} \subset S_a$, it holds that
\begin{equation}
\sum_{i \in [m]: y_i \notin T_{b}} w_i d(y_i,a) < \sum_{i \in [m]: y_i \notin T_{b}} w_i d(y_i,b) \leq \sum_{i \in [m]: y_i \in S_a} w_i d(y_i,b).
\label{eq:tree-71}
\end{equation}
Following the fact that
$$
\sum_{i \in [m]: y_i \in T_{b}} w_i d(y_i,a) \leq \sum_{i \in [m]: y_i \notin T_{b}} w_i d(y_i,a),
$$
and using Equations \eqref{eq:tree-70} and \eqref{eq:tree-71}, it follows that
$$
\sum_{i \in [m]: y_i \notin S_a} w_i d(y_i,b) < \sum_{i \in [m]: y_i \in S_a} w_i d(y_i,b),
$$
which contradicts necessity; the assertion follows.
\end{proof}

\begin{proof}
of Corollary~\ref{cor:wAvg-is-unique}:
Assume by way of contradiction that there exist two different weighted average locations, denoted by $a$ and $b$.
Following \lemref{lem:weighted-opt-derivative}, for every $T_j \in T(G,a)$, it holds that
$$
\sum_{i \in [m]: y_i \in T_j}  w_i d(y_i,a) \leq \sum_{i \in [m]: y_i \notin T_j}  w_i d(y_i,a).
$$
Following the same notation and reasoning as in the proof of \lemref{lem:weighted-opt-derivative} (sufficiency), it follows that
$$
\sum_{i \in [m]: y_i \notin S_a} w_i d(y_i,b) < \sum_{i \in [m]: y_i \in S_a} w_i d(y_i,b),
$$
which, according to \lemref{lem:weighted-opt-derivative}, implies that $b$ is not a weighted average location. The assertion follows.
\end{proof}

\begin{proof}
of \lemref{lem:movement-of-weighted-opt}:
Assume on the contrary that $d(a,a') > \sum_{i \in [m]}  w_i \delta_i$.
Let $T_0 \in T(G,a)$ be the subtree such that $a' \in T_0$,
and let $S_0 \in T(G,a')$ be the subtree such that $a \in S_0$.
By \lemref{lem:weighted-opt-derivative} we know that

$$
0 \leq \sum_{i \in [m]: y_i \notin T_0} w_i d(y_i,a) - \sum_{i \in [m]: y_i \in T_0} w_i d(y_i,a).
$$
It follows that
\begin{equation}
- \sum_{i \in [m]} w_i \delta_i \leq \sum_{i \in [m]: y'_i \notin T_0} w_i d(y'_i,a) - \sum_{i \in [m]: y'_i \in T_0} w_i d(y'_i,a).
\label{eq:tree11}
\end{equation}

We note that $T(G,a) \backslash T_0 \subset S_0$ and that $T(G,a') \backslash S_0 \subset T_0$. The next equations follow.
\begin{align}
\sum_{i \in [m]: y'_i \in S_0} w_i d(y'_i,a') = \sum_{i \in [m]: y'_i \notin T_0} w_i \left( d(y'_i,a) + d(a,a') \right) +
\sum_{i \in [m]: y'_i \in T_0 \cap S_0} w_i d(y'_i,a'), \label{eq:tree60} \\
\sum_{i \in [m]: y'_i \in T_0}  w_i d(y'_i,a) = \sum_{i \in [m]: y'_i \notin S_0} w_i \left( d(y'_i,a') + d(a,a') \right) +
\sum_{i \in [m]: y'_i \in T_0 \cap S_0} w_i d(y'_i,a). \label{eq:tree61}
\end{align}

Considering that
$$
\sum_{i \in [m]: y'_i \in T_0 \cap S_0} w_i d(y'_i,a') + \sum_{i \in [m]: y'_i \in T_0 \cap S_0} w_i d(y'_i,a) \geq \sum_{i \in [m]: y'_i \in T_0 \cap S_0}  w_i d(a',a),
$$
and combining Equations \eqref{eq:tree60} and \eqref{eq:tree61}, it follows that
\begin{align}
&\sum_{i \in [m]: y'_i \in S_0} w_i d(y'_i,a') - \sum_{i \in [m]: y'_i \notin S_0} w_i d(y'_i,a') &\geq \nonumber \\
&\sum_{i \in [m]: y'_i \notin T_0} w_i d(y'_i,a) - \sum_{i \in [m]: y'_i \in T_0}  w_i d(y'_i,a)  +
\sum_{i \in [m]} w_i d(a,a').&
\label{eq:tree62}
\end{align}
Since $d(a,a') > \sum_{i \in [m]}  w_i \delta_i$, it follows from Equations \eqref{eq:tree11} and \eqref{eq:tree62} that
$$
\sum_{i \in [m]: y'_i \in S_0} w_i d(y'_i,a') - \sum_{i \in [m]: y'_i \notin S_0} w_i d(y'_i,a') > 0.
$$
It then follows by \lemref{lem:weighted-opt-derivative} that $a'$ cannot be the weighted average location with respect to $G,\by'$ and $\bw$; hence we reach a contradiction.
\end{proof}

\begin{proof}
of \lemref{lem:differnce-of-costs-two-on-line}
For ease of presentation, we denote $\delta = d(a,b)$. When taking $sc(a,\bx)$ as a basis, the increase of cost when examining $sc(b,\bx)$ is as follows
\begin{align*}
&\sum_{i \in G \backslash T_b}(2d(x_i,a)\delta + \delta^2) + \sum_{i \in T_b}(2d(x_i,a)\delta - \delta^2) = \\
&|N|\delta^2 + 2\delta(\sum_{i \in G \backslash T_b}d(x_i,a) - \sum_{i \in T_b}d(x_i,a))
\end{align*}
The assertion follows.
\end{proof}

\begin{proof}
of \corref{cor:opt-flattening}
It is easy to see that $sc(a,\bx) = sc(a,\bx')$, that $sc(b,\bx) = sc(b,\bx')$, and that $\sum_{i \in G \backslash T_b}d(x_i,a) - \sum_{i \in T_b}d(x_i,a) = \sum_{i \in G \backslash T_b}d(x'_i,a) - \sum_{i \in T_b}d(x'_i,a)$. It is left to show that $\sum_{i \in G \backslash T_b}d(x'_i,a) - \sum_{i \in T_b}d(x'_i,a) = -|N|d(a,Opt')$.

It is clear that there are exactly two subtrees in $T(G,Opt')$ which contain agents (given $\bx'$). We denote them by $T_{Opt1}$, which contain $a$, and $T_{Opt2}$. Followed by \lemref{lem:weighted-opt-derivative}, it is easy to see that $\sum{x'_i \in T_{Opt1}}d(x'_i,Opt') = \sum{x'_i \in T_{Opt2}}d(x'_i,Opt')$.
It follows that
\begin{align*}
\sum_{i \in G \backslash T_b}(d(x'_i,a) + d(a,Opt')) + \sum_{i \in T_b \cap T_{Opt1}}(d(a,Opt') - d(x'_i,a)) = \sum_{i \in T_{Opt2}} (d(x'_i,a) - d(a,Opt')) \Leftrightarrow \\
\sum_{i \in G \backslash T_b}d(x'_i,a) - \sum_{i \in T_b}d(x'_i,a) = -|N|d(a,Opt').
\end{align*}
The assertion follows.
\end{proof}


\section{Missing proofs in \secref{sec:tree-general-mechanism}}
\label{sec:tree-general-mechanism-appendix}

\begin{proof}
of \thref{thm:tree-snc-is-truthful}:
Assume by way of contradiction that there exists an agent $i$ that can benefit by misreporting her location as $x_i'$, inducing a location profile $\bx' = (x'_i, x_{-i})$.
We quantify the effect of the deviation on the boomerang component and the average component of Mechanism \pb{}.
We begin with the boomerang component.
For every $j \in [m]$, let $\delta_j = d(f_j(\bx'),x_i) - d(f_j(\bx),x_i)$ be the additional cost incurred by $i$ due to the deviation, when $f_j$ is chosen.
Since $f_j$ is a boomerang mechanism, it holds that $d(f_j(\bx'),x_i) - d(f_j(\bx),x_i) = d(f_j(\bx'),f_j(\bx)) \geq 0$.
Therefore, the additional cost incurred by $i$ due to the boomerang component is $\sum_{j \in [m]}  w_j \delta_j = \sum_{j \in [m]} w_j d(f_j(\bx'),f_j(\bx))$.
By \lemref{lem:movement-of-weighted-opt}, the average component reduces agent $i$'s cost by at most $\sum_{j \in [m]} w_j d(f_j(\bx'),f_j(\bx))$.
The assertion of the theorem follows.
\end{proof}

\begin{proof}
of \propref{prop:exmaples}:
In what follows we establish the correctness of the assertions of the proposition.
\begin{enumerate}
\item Any mechanism that chooses the $k$-location, is an instance of Mechanism \pb{}, as it is a boomerang mechanism.
\item The left-right-middle (LRM) mechanism is an instance of Mechanism \pb{}, with $m=2$, $\bw=(1/2, 1/2)$, and $f_1$ and $f_2$ being the left-most and right-most mechanisms, respectively (which are boomerang mechanisms). The weighted average location is the center of the left-most and right-most locations, and is chosen with probability $1/2$, while each of the extreme points is chosen with probability $1/4$.
\item The random dictator (RD) mechanism is the uniform distribution over dictator mechanisms, which are boomerang mechanisms. The assertion follows by \corref{cor:snc-dist-is-truthful}.
\end{enumerate}
\end{proof}

\begin{example}
\label{exm:counter-necessary}
Consider the following mechanism on a line. Choose each of the leftmost and rightmost agents with probability $\frac{1}{2n}$, and the center of every two consecutive locations with probability $\frac{1}{n}$. It is easy to verify that this is an SP mechanism, yet it cannot be formulated as a special case of Mechanism \pb{} or a probability distribution over \pb{} mechanisms.
\end{example}


\section{Missing proofs in \secref{sec:SP-mechanisms-line}}
\label{sec:SP-mechanisms-line-appendix}

\begin{proof}
of Claim~\ref{lem:optimal-point}:
The optimal location can be derived by taking the derivative of the SOS cost function.
As the derivative of $\sum_{i \in N} |y-x_i|^2$ is $\sum_{i \in N} 2y-2x_i$, the point that minimizes the SOS cost is $\frac{\sum_{i \in N} x_i}{n}$.
\end{proof}

\begin{proof}
of \lemref{lem:immigrants}:
In the profile $\bx^0$, $m$ agents are located at $c$.
For $j = 0, \ldots, m$, let $\bx^j$ be the profile in which $n-m$ agents are located at $a$, $j$ agents are located at $b$, and $m-j$ agents are located at $c$; and let $y^j \sim f(\bx^j)$.
Since $f$ is SP, it must hold that $E[|c-y^0|] \leq E[|c-y^1|]$; otherwise, in the profile $\bx^0$, every agent can decrease her cost by misreporting her location to be at $b$ instead of $c$.
Similarly, for every $j = 2, \ldots, m$, it holds that $E[|c-y^{j-1}|] \leq E[|c-y^j|]$.
It follows that $E[|c-y^0|] \leq E[|c-y^m|]$.
The same line of arguments is used to prove that $E[|b-y^m|] \leq E[|b-y^0|]$, which concludes the proof.
\end{proof}

\begin{proof}
of Theorem~\ref{thm:median-two-approx}
Let $\mu$ be a median location in $\bx$. It is well known that the median mechanism on a line is truthful (see, e.g., \cite{PT09}).
Therefore, it remains to show that it provides a $2$-approximation to the miniSOS objective; formally, we need to show that
\begin{equation}
\sum_{i \in N} |x_i-\mu|^2 \leq 2\sum_{i \in N} |Opt-x_i|^2.
\label{eq:eq1}
\end{equation}

Assume without loss of generality that $x_1 \leq x_2 \leq \ldots \leq x_n$.
Assume additionally that $\mu \leq Opt$ (the proof works analogously for the case in which $\mu \geq Opt$).
It is easy to verify that
\begin{align*}
\sum_{i \in N} |x_i-\mu|^2 &= \sum_{i \in N} |(x_i-Opt)+(Opt-\mu)|^2 \\
&= \sum_{i \in N} \left( (x_i-Opt)^2 + 2(Opt - \mu)(x_i - \frac{Opt + \mu}{2}) \right).
\end{align*}
Therefore, by subtracting $\sum_{i \in N} |Opt-x_i|^2$ from both sides, it remains to prove that
$$
\sum_{i \in N}  2(Opt - \mu)(x_i - \frac{Opt + \mu}{2}) \leq \sum_{i \in N} |Opt-x_i|^2,
$$
which is equivalent to showing that
\begin{align*}
&\sum_{i \leq \lceil n/2 \rceil} 2(Opt - \mu) \left((x_{i} - \frac{Opt + \mu}{2})+(x_{n+1-i} - \frac{Opt + \mu}{2}) \right) \leq \\
&\sum_{i \leq \lceil n/2 \rceil} |Opt-x_i|^2 + |Opt-x_{n+1-i}|^2.
\end{align*}
We next show that the last inequality holds piecewise; i.e., for every $i \leq \lceil n/2 \rceil$, it holds that
$$
2(Opt - \mu) \left( (x_{i} - \frac{Opt + \mu}{2})+(x_{n+1-i} - \frac{Opt + \mu}{2}) \right) \leq  (Opt-x_{i})^2 + (Opt-x_{n+1-i})^2.
$$
For every $i \leq \lceil n/2 \rceil$, it holds that $x_i \leq \mu \leq Opt$; thus $(Opt-\mu)^2 \leq (Opt-x_{i})^2$, and it suffices to show that
$$
2(Opt - \mu) \left( \mu - \frac{Opt+\mu}{2}+ x_{n+1-i} - \frac{Opt + \mu}{2}) \right) \leq  (Opt-\mu)^2 + (Opt-x_{n+1-i})^2.
$$
It can be easily verified that the above inequality holds iff $(2Opt - \mu - x_{n+1-i})^2 \geq 0$; the assertion of the theorem follows.
\end{proof}

\begin{proof}
of Theorem~\ref{thm:deterministic-bound}:
Assume by way of contradiction that there exists a deterministic SP mechanism $f$ which yields a better approximation than $2$.
Consider a location profile $\bx$, in which $\frac{n}{2}$ agents are located at $0$, and $\frac{n}{2}$ agents are located at $2$, and let $f(\bx)=p$.
Simple calculations show that to achieve a better approximation than $2$, it must hold that $p \in (0,2)$ (note that the optimal location is $1$, which obtains an SOS cost of $n$, while the locations $0$ or $2$ obtains each an SOS cost of $2n$).
Now consider a different location profile, denoted $\bx^p$, in which $\frac{n}{2}$ agents are located at $0$, and $\frac{n}{2}$ agents are located at $p$.
Following the same argument, it must hold that $f(\bx^p) \in (0,p)$, and thus $|p-f(\bx^p)|>0$.
Since $f(\bx)=p$, we get that $|p-f(\bx^p)|>|p-f(\bx)|$, which implies, by \lemref{lem:immigrants}, that $f$ is not SP.
\end{proof}

\begin{proof}
of Theorem~\ref{thm:random-2-approx}:
Let $rd(\bx)$ denote the RD mechanism.
It holds that
\begin{align*}
sc(rd(\bx),\bx)  &= \sum_{i \in N} \frac{1}{n} \sum_{j \in N} \left( x_i - x_j \right) ^2 &\mbox{(by def. of sc)}\\
            &= \sum_{i \in N} x_i^2 +  \sum_{j \in N} x_j^2 - \frac{2}{n}\sum_{j \in N}\sum_{i \in N}x_i x_j \\
            &= 2 \sum_{i \in N} x_i^2 - \frac{2}{n}\sum_{j \in N}x_j \sum_{i \in N} \left(2x_i - \frac{\sum_{j \in N} x_j}{n} \right) &\mbox{(by $\sum_{i \in N} \left( x_i - \frac{\sum_{j \in N} x_j}{n} \right) = 0$)}\\
            &= 2 \sum_{i \in N} x_i^2 - 2Opt \sum_{i \in N} \left( 2x_i - Opt \right) &\mbox{(by \obsref{lem:optimal-point})}\\
            &= 2 \sum_{i \in N} (Opt - x_i)^2 \\
            &= 2 sc(Opt,\bx) &\mbox{(by def. of sc)}.
\end{align*}
\end{proof}

\begin{proof}
of Theorem~\ref{thm:random-1.5-approx}:
We first prove the approximation factor.
Let $g(\bx)$ denote mechanism~\ref{mech:half_opt_half_RD}, and let $avg(\bx)$ denote the average point.
By \obsref{lem:optimal-point}, the optimal location with respect to miniSOS is $avg(\bx)$.
We get
$$
\frac{sc(g(\bx),\bx)}{Opt} = \frac{\half RD(\bx) + \half avg(\bx)}{avg(\bx)} = 1.5,
$$
where the last equation follows by Theorem~\ref{thm:random-2-approx}.

In order to show the strategyproofness of the mechanism, it suffices to prove that it is an instance of Mechanism \pb{}.
Indeed, one can easily verify that this is a special case in which $m=n$, $\bw$ is the uniform distribution over $[n]$, and for
every $i \in [n]$, $f_i(\bx) = x_i$ (i.e., $f_i$ is dictatorship with agent $i$ as the dictator).
\end{proof}

\begin{proof}
of Theorem~\ref{thm:random-bound}:
Assume on the contrary that there exists an SP mechanism $f(\bx)$ and $\epsilon>0$ such that $f(\bx)$ yields an approximation ratio of $1.5-\epsilon$, and let $y \sim f(x)$ (i.e., $y$ is a random variable that is distributed according to $f(x)$).
We shall use the following lemma in the proof.
\begin{lemma}
\label{lem:sum-distances-is-3}
There exists some $a \in \reals$ and a location profile $\bx$ in which $\frac{n}{2}$ agents are located at $a$ and $\frac{n}{2}$ agents are located at $a+4$, such that $E[|y-(a+1)| + |y-(a+3)|]>3-2\epsilon$.
\end{lemma}

\begin{proof}
Consider a location profile $\bx^0$, in which $\frac{n}{2}$ agents are located at $0$ and $\frac{n}{2}$ agents are located at $4$, and let $y^0 \sim f(\bx^0)$.
If $E[|y^0-3|+|y^0-1|]>3-2\epsilon$, then we are done.
Otherwise, either $E[|y^0-1|] \leq 1.5-\epsilon$ or $E[|y^0-3|] \leq 1.5-\epsilon$.
Assume w.l.o.g. that the latter holds, and consider a location profile $\bxbar^{0}$, in which $\frac{n}{2}$ agents are located at $0$ and the rest are located at $3$.
Let $\ybar^{0} \sim f(\bxbar^{0})$.
By \lemref{lem:immigrants}, to preserve truthfulness, it must hold that $E[|\ybar^{0}-3|] \leq 1.5-\epsilon$, which implies that $E[|\ybar^{0}-0|] \geq 1.5+\epsilon$.

Consider next a location profile $\bx^1$, in which $\frac{n}{2}$ agents are located at $-1$ and the rest are located at $3$, and let $y^1 \sim f(\bx^1)$.
By \lemref{lem:immigrants}, to preserve truthfulness, it must hold that $E[|y^1-0|] \geq 1.5+\epsilon$.
If $E[|y^1-2|+|y^1-0|]>3-2\epsilon$, then we are done. Otherwise, it follows that $E[|y^1-2|] \leq 1.5-3\epsilon$.

We continue iterating such that in iteration $j=1, 2, \ldots$, a profile $\bx^j$ is considered, in which half the agents are located at $-j$ and half are located at $4-j$, and for every profile $\bx^j$, we denote by $y^j$ the random variable distributed according to $f(\bx^j)$.
We show that there exists some $j$ for which $E[|y^j-(1-j)|+|y^j-(3-j)|]>3-2\epsilon$; the assertion of the lemma then follows by substituting $a = -j$.
It remains to prove the last inequality.
Indeed, by repeatedly applying \lemref{lem:immigrants} for every $j$, we get that
\begin{equation}
\label{eq:101}
\mbox{ if } E[|y^j-(1-j)|+|y^j-(3-j)|] \leq 3-2\epsilon, \mbox{ then } E[|y^j-(3-j)|] \leq 1.5 - \epsilon(2j+1).
\end{equation}
But since for $j > \frac{1.5 - \epsilon}{2\epsilon}$, it holds that $1.5 - \epsilon(2j+1) < 0$, it must hold that $E[|y^j-(3-j)|] > 1.5 - \epsilon(2j+1)$, which, by Equation~\ref{eq:101} implies that $E[|y^j-(1-j)|+|y^j-(3-j)|] > 3-2\epsilon$.
It follows that the profile $\bx^j$ satisfies the conditions of the lemma, and the proof follows.
\end{proof}

With this lemma, we are ready to prove the theorem.
Let $\bx$ be a location profile that satisfies the conditions of \lemref{lem:sum-distances-is-3}, and assume w.l.o.g. that $\bx$ is a profile in which half the agents are located at $0$ and half at $4$.
By the last lemma, it holds that
\begin{equation}
\label{eq:high_expectation}
E[|y-1|+|y-3|]>3-2\epsilon,
\end{equation}
where $y \sim f(\bx)$.
For ease of presentation, let $p=Pr(|y-2| \leq 1)$ and let $z=E[|y-2|:|y-2|>1]$.
It holds that
\begin{align*}
E[|y-1|+|y-3|] &= E \left[ |y-1|+|y-3|:|y-2|>1 \right](1-p) \\
& + E \left[ |y-1|+|y-3|:|y-2| \leq 1 \right]p \\
&= 2z(1-p) + 2p.
\end{align*}
Therefore, by Equation (\ref{eq:high_expectation}), it follows that $2z(1-p) + 2p > 3-2\epsilon$.
Since $z>1$ by definition, it follows that
\begin{equation}
\label{eq:small_prob}
p < 1 - \frac{1-2\epsilon}{2z-2}.
\end{equation}

We now turn to calculate the SOS cost of the profile $\bx$ induced by the mechanism $f(\bx)$.
It holds that
\begin{align}
E[sc(y,\bx)] &= E[sc(y,\bx):|y-2| \leq 1]p + E[sc(y,\bx):|y-2| > 1](1-p) \nonumber \\
 &\geq sc(E[y:|y-2| \leq 1],\bx)p + E[sc(y,\bx):|y-2| > 1](1-p),
\label{eq:eq20}
\end{align}
where the last inequality follows from Jensen's inequality.
Since $2$ is the optimal location in the profile $\bx$, it holds that
\begin{equation}
sc(E[y:|y-2| \leq 1],\bx) \geq sc(2,\bx) = 4n.
\label{eq:eq21}
\end{equation}
We also have that
\begin{align}
E[sc(y,\bx):|y-2| > 1] & = E[\frac{n}{2}(y^2 + (y-4)^2):|y-2| > 1] \nonumber \\
&= \frac{n}{2}(8 + 2E[(y-2)^2:|y-2| > 1]) \nonumber \\
&\geq \frac{n}{2}(8 + 2E^2[|y-2|:|y-2| > 1]) \nonumber \\
&= \frac{n}{2}(8+2z^2), \label{eq:eq22}
\end{align}
where the last inequality follows from Jensen's inequality.
By Substituting \eqref{eq:eq21} and \eqref{eq:eq22} in \eqref{eq:eq20}, we get that $E[sc(y,\bx)] \geq \frac{n}{2}(8 + 2z^2 - 2z^2p)$.
It, therefore, follows by \eqref{eq:small_prob} that
\begin{equation}
E[sc(y,\bx)] > \frac{n}{2}(8 + 2z^2 - 2z^2(1 - \frac{1-2\epsilon}{2z-2})) = 4n  + \frac{nz^2(1-2\epsilon)}{2z-2}.
\label{eq:eq42}
\end{equation}
It is easy to verify that this function attains its minimum at $z=2$, with a value of $6n - 4n\epsilon$.
Thus, $E[sc(y,\bx)] > 6n - 4n\epsilon$.
The optimal solution is to locate the facility at $2$, which yields an SOS cost of $4n$.
We get that $\frac{E[sc(y,\bx)]}{Opt} > \frac{6n - 4n\epsilon}{4n} = 1.5 - \epsilon$, and a contradiction is reached.
The assertion of the theorem follows. \qedd
\end{proof}


\section{Missing proofs in \secref{sec:tree}}
\label{sec:tree-appendix}

\begin{proof}
of \thref{thm:tree-det-mechanism}:
It has been shown in earlier papers that a median of a tree is SP (see, e.g., \cite{Alon09}).
Here we show that it provides $2$-approximation with respect to the miniSOS objective.

We start with some notation. Given a location profile $\bx$, let $\mu$ and $Opt$ denote the median and the optimal location, respectively.
Let $T_{\mu}^{Opt} \in T(G,\mu)$ denote the subtree of $T(G,\mu)$ that contains $Opt$, and
let $L(G,\bx) \subseteq N$ denote the subset of agents $i$ such that $x_i \not\in T_{\mu}^{Opt} \setminus \{\mu\}$.
Similarly, let $T_{Opt}^{\mu} \in T(G,Opt)$ denote the subtree of $T(G,Opt)$ that contains $\mu$, and
let $R(G,\bx) \subseteq N$ denote the subset of agents $i$ such that $x_i \not\in T_{Opt}^{\mu} \setminus \{Opt\}$.
Finally, let $M(G,\bx) \subseteq N$ denote the subset of agents that belong to neither $L$ nor $R$.
These are the agents $i$, such that $x_i \neq \mu$ and $x_i \neq Opt$, for which $path(x_i,\mu) \cap path(x_i,Opt)$ is one location on the graph, rather than a path.
When clear in the context, we simplify notation and write $L,R,M$.

With this, we continue. Suppose $\mu \neq Opt$ (if $\mu=Opt$, we are done).
Let $k$ be the number of agents in $M(G,\bx)$, and refer to these agents as $x_1, \ldots, x_k$ (order them arbitrarily).
In what follows we describe an iterative process, which starts with a tuple $(G^0=G, \bx^0=\bx)$ of graph and location profile, and in every iteration $j \in [k]$, a new tuple $(G^j,\bx^j)$ is induced in a specific manner that will be specified soon.
We will show that the approximation ratio in every such iteration can only get worse.
The assertion of the theorem will then be established by proving the desired approximation ratio on the final instance.

For every $j \in [k]$, let $Opt^j$ and $\mu^j$ denote the respective optimal location and median with respect to $(G^j,\bx^j)$.
For every $j \in [k]$, define $(G^{j},\bx^{j})$ as follows.
Starting with $G^{j-1}$, add an edge, rooted at $Opt^{j-1}$, of length $d_G \left(x_j,Opt(G,\bx) \right)$, and move $x_j$ from its position in $G^{j-1}$ to the end of this edge (see Figure~\ref{fig:median-process} for illustration).
We will also use the notation $x_i^j$ to denote the location of $x_i$ in $G^{j}$ for every $i \in [n], j \in [k]$.

It is easy to see that $\mu^j = \mu^{j-1}$ for every $j$ (since agents never cross the median in their new location).
We shall, therefore, use $\mu$ to denote the median of every iteration.

It is slightly more subtle to see that the optimal point does not change either; i.e., $Opt^j = Opt^{j-1}$ for every $j \in [k]$.

\begin{lemma}
\label{lem:opt-did-not-change}
For every $j \in [k]$, it holds that $Opt^{j-1} = Opt^{j}$.
\end{lemma}
\begin{proof}
Let $T_{\mu}$ denote the subtree of $T(G^{j-1},Opt^{j-1})$ that contains $\mu$.
Note that in $G^{j-1}$, $x_j \in T_{\mu}$ as well.
Let $T'_{\mu}$ denote the subtree of $T(G^{j},Opt^{j-1})$ that contains $\mu$, and
let $T_j$ denote the subtree of $T(G^{j},Opt^{j-1})$ that contains $x_j$.
Note that $|T(G^{j},Opt^{j-1})| = |T(G^{j-1},Opt^{j-1})| + 1$.

By applying \lemref{lem:weighted-opt-derivative} with $\bw$ as the uniform distribution, for every $T \in T(G^{j-1},Opt^{j-1})$, it holds that
$$
\sum_{i \in [n]: x^{j-1}_i \in T} d(x^{j-1}_i,Opt^{j-1}) \leq \sum_{i \in [n]: x^{j-1}_i \notin T} d(x^{j-1}_i,Opt^{j-1}).
$$
It is not too difficult to see that this implies that for every $T' \in T(G^{j},Opt^{j-1})$ such that $T' \neq T'_\mu, T_j$, it holds that
$$
\sum_{i \in [n]: x^j_i \in T'} d(x^j_i,Opt^{j-1}) \leq \sum_{i \in [n]: x^j_i \notin T'} d(x^j_i,Opt^{j-1}).
$$
In addition, due to the fact that we split $T_\mu$ into $T'_\mu$ and $T_j$, the same holds for them.
It follows from \lemref{lem:weighted-opt-derivative} that $Opt^{j-1} = Opt^j$.
\end{proof}

Since this iterative process changed neither the optimal location, nor the distance of the agents from the optimal location, the optimal social cost did not change either.
That is, the optimal social cost in $(G,\bx)$ is the same as the optimal social cost in $(G^k,\bx^k)$.
In contrast to the optimal cost, it is not difficult to see that the mechanism's cost increased.
This is because the distance between $\mu$ and the agents in $M(G,\bx)$ increased (while it did not change for the other agents).
It follows that it is sufficient to prove the desired approximation ratio on the instance $(G^k,\bx^k)$.

In order to establish a $2$-approximation on $(G^k,\bx^k)$, we modify the instance $(G^k,\bx^k)$ to a line instance $(\reals,\bz)$ (for $\bz \in \reals^n$), such that, the distances of every agent from the median and from the optimal location are preserved.


Let $L^k=L(G^k,\bx^k), R^k=R(G^k,\bx^k)$ and $M^k=M(G^k,\bx^k)$.
By construction, $M^k = \emptyset$.
Consider the following location profile, $ ( z_1, \ldots, z_n )$, on a line.
For every $i \in L^k$, let $z_i$ be located at $-d_{G^k}(x^k_i,\mu^k)$.
For every $i \in R^k$, let $z_i$ be located at $d_{G^k}(\mu^k,x_i^k)$ (which is equal to $d_{G^k}(\mu^k,Opt^k) + d_{G^k}(x^k_i,Opt^k)$).
Let $b = d_{G^k}(\mu^k,Opt^k)$.
For every $i \in [n]$ it holds that $d(z_i,0)= d_{G^k}(x^k_i,\mu^k)$, and $d(z_i,b)= d_{G^k}(x^k_i,Opt^k)$; i.e., the respective distances from $\mu^k$ and $Opt^k$ were preserved.

Let $\mu_{\bz}$ and $Opt_{\bz}$ denote the median and the optimal location with respect to $(\reals, \bz)$, respectively.
It is easy to verify that $\mu_{\bz} \leq 0$ and $Opt_{\bz} \geq b$.
By \thref{thm:median-two-approx}, locating the facility at $\mu_{\bz}$ provides $2$-approximation for this line instance.
But since $\mu_{\bz} \leq 0 < b \leq Opt_{\bz}$, it follows that locating the facility at $0$ yields a social cost (SOS) that is at most twice that of $b$; i.e.,
$$
\sum_{i \in [n]} d(z_i,0)^2 \leq 2 \sum_{i \in [n]} d(z_i,b)^2.
$$
However, because the respective distances from $\mu^k$ and $Opt^k$ were preserved, it follows that
$$\sum_{i \in [n]} d_{G^k}(x^k_i,\mu^k)^2 \leq 2 \sum_{i \in [n]} d_{G^k}(x^k_i,Opt^k)^2,
$$
as required; the assertion follows.
\end{proof}

\begin{proof}
of \propref{prop:PC-is-boomerang}:
It is easy to verify that an agent can modify the location of the mechanism's outcome only by declaring herself to be on the location's opposite side, thus pushing the returned facility location away from its true location.
It follows directly that the distance between the original location and the new one is exactly the additional cost imposed on a misreporting agent.
\end{proof}

\begin{proof}
of Lemma~\ref{lem:projected-path}:
For every $i \in [n]$, let $y_i = f_i(\bx)$.
By design, every subtree $T \in T(G,y_j)$ (excluding $y_j$) contains less than $qn$ agents.
Assume by way of contradiction that the points $y_1, \ldots, y_n$ are not located on a single path.
Let $T_y$ be the subtree induced by connecting all the $y_i$'s (i.e., for any $a \in T_y$, there exist $i,j \in [n]$ such that $a \in path(y_i,y_j)$.
By the contradiction assumption, $T_y$ contains at least three leaves; assume w.l.o.g. these leaves are $y_1, y_2, y_3$.
Then, for every $j \in \{1,2,3\}$, there exists a subtree in $T(G,y_j)$, call it $T^j$, that contains $T_y$.
By design, $T^j \setminus \{y_j\}$ contains less than $qn$ agents;
therefore, more than $(1-q)n$ agents are located at $\{ T(G,y_j) \setminus T^j \} \cup \{y_j\}$.
It is easy to verify that $\{ T(G,y_1) \setminus T^1 \} \cup \{y_1\} \subset T^3 \setminus y_3$ and also $\{ T(G,y_2) \setminus T^2 \} \cup \{y_2\} \subset T^3 \setminus y_3$.
We get that $T^3 \setminus y_3$ contains more than $(2-2q)n \geq \frac{2}{3} n$ agents, where the last inequality follows by $\frac{1}{2} < q \leq \frac{2}{3}$.
This contradicts the choice of $y_3$ by $f_3(\bx)$; the assertion follows.
\end{proof}

\begin{proof}
of \lemref{lem:y-is-one-location}:
Assume that the optimal location is not at $y_1$ (otherwise, the approximation ratio is $1$ and we are done).
By design, if $y_1=y_2=\ldots=y_n$, then there does not exist $T \in T(G,y_1)$ such that there are more than $n/3$ agents in $T \setminus \{ y_1 \}$.
Let $T_{Opt} \in T(G,y_1)$ such that the optimal location is in $T_{Opt}$. Let $G'$ be a new tree which is created by positioning a new subtree $T'$ which is constructed by a long edge, rooted at $y_1$, and let $\bx'$ be a new location profile, such that for any $x_i \in T_{Opt} \setminus \{ y_1 \}$, $x'_i$ is located on $T'$, at the location that is distanced $d(y_1,x_i)$ from $y_1$.
Let the optimal location with respect to the new tree and $\bx'$ be denoted as $Opt'$. Let $\delta$ denote $d(Opt',y_1)$.
It is straightforward that because $Opt$ is in $T_{Opt}$, then $Opt'$ is in $T'$. Let $T_R \in T(G',Opt')$ such that $y_1 \in T_R$ and let $T_L$ be the other subtree.

Following a simple variation of \corref{cor:opt-flattening}, it holds that
\begin{equation}
\label{eq:tree-600}
sc_{G'}(y_1,\bx') = sc_{G'}(Opt',\bx') + n\delta^2.
\end{equation}

By \lemref{lem:weighted-opt-derivative} it follows that $\sum_{x'_i \in T_L}d(x'_i,Opt') = \sum_{x'_i \in T_R}d(x'_i,Opt')$. Since the right expression is at least $\frac{2}{3}n\delta$, so is the left one.

We shall next examine the social cost of $Opt'$ with respect to $G'$ and $\bx'$. at least $\frac{2}{3}n$ agents in $T_R$ add at least $\delta^2$ each to the social cost obtained at $Opt'$. On the other hand, following Jensen's inequality, the minimal addition agents in $T_L$ yield to the social cost is obtained when they are located at distance of $\frac{\sum_{x'_i \in T_L}d(x'_i,Opt')}{|T_L|}$ from $Opt'$. Thus, their addition is at least
$$
|T_L| \left( \frac{\sum_{x'_i \in T_L}d(x'_i,Opt')}{|T_L|} \right)^2 \geq |T_L| \left( \frac{\frac{2}{3}n\delta}{|T_L|} \right)^2.
$$
As this expression decreases in $|T_L|$, and noting that $|T_L| \leq \frac{n}{3}$, it follows that their addition is at least $\frac{4}{3}n\delta^2$.
Summing it all together, it follows that the $sc_{G'}(Opt',\bx') \geq 2n\delta^2$. Applying \eqref{eq:tree-600}, it follows that
$$
\frac{sc_{G'}(y_1,\bx')}{sc_{G'}(Opt',\bx')} \leq 1.5.
$$

Now, when considering $G$ and $\bx$, we see that the mechanism's induced social cost does not change (as the distances from $y_1$ do not change), while we added constraints that might affect the optimal location. The assertion follows.
\end{proof}

We shall now present the formal statements of the lemmata that are required in the proof of Theorem~\ref{thm:rndom-tree-1.82}.
Since $y_1$ and $y_n$ are symmetric, we state and prove every lemma for one case; the complementary one can be proved analogously.


\begin{lemma}
\label{lem:flattening}
Let $i$ be an agent such that $y_i$ is in the open $path(y_1,y_n)$, and assume that $\bar{Opt} \in path (y_i,y_n)$.
Let $(G',\bx')$ be a transformation of $(G,\bx)$ as follows. If $d(x_i,y_i) \leq d(y_1,y_i)$, then position $i$ on $path(y_1,y_i)$, distanced $d(x_i,y_i)$ from $y_i$. If $d(x_i,y_i) > d(y_1,y_i)$, then create an edge, rooted at $y_1$, of length $d(x_i,y_i) - d(y_1,y_i)$, and locate $i$ at its tip.
Let $(G'',\bx'')$ be a transformation of $(G,\bx)$ as follows. Create an edge, rooted at $\bar{Opt}$, of length $d(x_i,\bar{Opt})$, and place $i$ at its tip. Then, the mechanism obtains a worse approximation ratio for at least one of the tuples $(G',\bx')$ and $(G'',\bx'')$, compared to $(G,\bx)$.
\end{lemma}

\begin{proof}
We start by defining a parameterized instance, namely $(G^{\delta},\bx^{\delta})$, as follows. For any $0 \leq \delta \leq d(y'_i, \bar{Opt})$, originating at $(G',\bx')$, create an edge rooted at the location that is distanced $\delta$ from $y'_i$ on $path(y'_i,\bar{Opt})$, of length $d(x'_i,y'_i) + \delta$, and locate $i$ at its tip. We note that $(G^0,\bx^0)$ and $(G^{d(y^0_i, \bar{Opt})},\bx^{d(y^0_i, \bar{Opt})})$ are in fact $(G',\bx')$ and $(G'',\bx'')$, respectively.

Note that the optimal cost for any $\delta$ does not change, as the distance between agent $i$ and the original optimal location does not change, and the other agents are left untouched. Therefore, it is left to show that the social cost obtained by the mechanism is higher for either $\delta = 0$ or for $\delta =  d(y^0_i, \bar{Opt})$ than the social cost obtained for the original instance ($G,\bx)$.

We shall use the following notations.
Let $L_i$ denote the agents that do not share the same subtree of $y^0_i$ as $y_n$, i.e., let $T_n \in T(G^0,y^0_i)$ such that $y_n \in T_n$; then, $j \in L_i$ iff $x^0_j \notin T_n$.
Given $\delta$, let $M^\delta_i$ be defined as follows. $j \in M^\delta_i$ iff $j \neq i$ and $y_j \in path(y^0_i,y^{\delta}_i)$. Let $R^\delta_i$ be the complementary set of the agents excluding agent $i$, i.e., $R^\delta_i = N_{-i} \backslash \{L_i \cup M^\delta_i \}$.
Let $L_{avg(\by^0)}$ denote the agents that do not share the same subtree of $avg(\by^0)$ as $y_n$, i.e., let $T_n \in T(G^0,avg(\by^0))$ such that $y_n \in T_n$, then $j \in L_{avg(\by^0)}$ iff $x^0_j \notin T_n$. Similarly, we shall define $R_{avg(\by^0)} = N \backslash L_{avg(\by^0)}$.

Let $\sigma = min(d(y_1,y_i),d(x_i,y_i))$. We note that $G$ is in fact $G^{\sigma}$. Let $h(\delta)$ represent the social cost obtained by the mechanism for the instance $(G^{\delta},\bx^{\delta})$. We shall next show how $h(\delta)$ differs from $h(0)$, i.e., we shall characterize $g(\delta)$ such that $h(\delta) = h(0) + g(\delta)$.

As we come to analyze $g(\delta)$, we can divide the analysis into two components, namely the agent component and the weighted average component.

We shall start with the agent component. If the projected location of an agent $j \in R^\delta_i$ is chosen, then the obtained social cost does not change, as the distance of $x_i$ from $y_j$ does not change. If the projected location of an agent $j \in L_i$ is chosen (i.e., $y_j$), then the social cost increases (with respect to $G^0$ and $\bx^0$) by

\begin{equation}
\label{eq:100}
4d(y_j,x^0_i)\delta +4\delta^2.
\end{equation}
If the projected location of an agent $j \in M^\delta_i$ is chosen, then the social cost increases by

\begin{equation}
\label{eq:big2}
4d(y_j,x^0_i)(\delta - d(y_j,y^0_i)) + 4(\delta - d(y_j,y^0_i))^2.
\end{equation}
If the projected location of agent $i$ is chosen, then the social cost increases by

\begin{equation}
\label{eq:big3}
\sum_{j \in L_i} \left( 2d(y^0_i,x_j)\delta + \delta^2 \right) -  \sum_{j \in \{M^\delta_i \cup R^\delta_i\} } \left( 2d(y^0_i,x_j)\delta - \delta^2 \right) + \delta^2 ,
\end{equation}
where the last summand is for $i$ herself.
Every agent's projected location is chosen with probability $\frac{1}{2n}$, thereby deriving the agents component in $f(\delta)$.

We continue with the weighted average component.
The influence obtained by the weighted average component (i.e., by $avg(\by)$), consists of two parts; namely, the movement of $avg(\by)$ itself, and the influence of the movement of $i$ on the social price obtained by $avg(\by)$.
For the first part,  since $avg(\by)$ smoothly slides towards $y_n$ by the distance of $\frac{\delta}{n}$, the difference in cost originating by the first part is given by

\begin{equation}
\label{eq:big4}
\sum_{j \in L_{avg(\by^0)}} \left( 2d(x^0_j,avg(\by))\frac{\delta}{n} + \left( \frac{\delta}{n} \right) ^2 \right) -  \sum_{j \in R_{avg(\by^0)}} \left( 2d(x^0_j,avg(\by))\frac{\delta}{n} - \left( \frac{\delta}{n} \right) ^2 \right).
\end{equation}
For the second part, similarly to the classification done when characterizing the agents' component, we need to distinguish between three scenarios. Denote this influence by $\beta$.

If $avg(\by^\delta)$ is not on the same subtree of $y^0_i$ as $y_n$, i.e., let $T_n \in T(G^0,y^0_i)$ such that $y_n \in T_n$, and $avg(\by^\delta) \notin T_n$, then $\beta = 4d(avg(\by^\delta),x^0_i)\delta + 4\delta^2$.
If $avg(\by^\delta) \in path(y^0_i,y^{\delta}_i)$, then $\beta = 4d(avg(\by^\delta),x^0_i)(\delta - d(avg(\by^\delta),y^0_i)) + 4(\delta - d(avg(\by^\delta),y^0_i))^2$.
If neither of the above holds, then it preserves its distance from $x_i$, i.e., $\beta = 0$.

In summary, when summing over the different components, we get:
\begin{align*}
g(\delta) = &\frac{1}{2n} \left(\sum_{j \in L_i} \left( 4d(y_j,x^0_i)\delta +4\delta^2 \right)   \right) + \\
&\frac{1}{2n} \left(\sum_{j \in M^\delta_i} \left(  4d(y_j,x^0_i)(\delta - d(y_j,y^0_i)) + 4(\delta - d(y_j,y^0_i))^2  \right)   \right) + \\
&\frac{1}{2n} \left ( \sum_{j \in L_i} \left( 2d(y^0_i,x_j)\delta + \delta^2 \right) -  \sum_{j \in \{M^\delta_i \cup R^\delta_i\} } \left( 2d(y^0_i,x_j)\delta - \delta^2 \right) + \delta^2 \right ) + \\
&\frac{1}{2} \left(  \sum_{j \in L_{avg(\by^0)}} \left( 2d(x^0_j,avg(\by))\frac{\delta}{n} + \left( \frac{\delta}{n} \right) ^2 \right) -  \sum_{j \in R_{avg(\by^0)}} \left( 2d(x^0_j,avg(\by))\frac{\delta}{n} - \left( \frac{\delta}{n} \right) ^2 \right) + \beta \right).
\end{align*}
One can verify that $g(\delta)$ is continuous, and is piecewise differentiable (in particular, it is not differentiable for any $\delta$ such that there exists $j \in N \backslash L_i$ such that $\delta = d(y^0_i,y_j)$, and when $\delta$ is such that $avg(\by^\delta) = y^{\delta}_i$). Moreover, one can also verify that for any $\delta$, the left-hand derivative is smaller or equal to the right-hand derivative of $g$, i.e., for any $\delta$, $g'_{-}(\delta) \leq g'_{+}(\delta)$. In other words, $g(\delta)$'s derivative is monotonically increasing where it is defined (this can be easily seen, as the quadratic components are all positive). Recalling that $h(\delta) = h(0) + g(\delta)$, and considering that $0 \leq \sigma \leq d(y^0_i, \bar{Opt})$, it follows that either $h(0) \geq h(\sigma)$, or $h(d(y^0_i, \bar{Opt})) \geq h(\sigma)$. The fact that $(G^\sigma,\bx^\sigma)$ is $(G,\bx)$ concludes the proof.

\end{proof}


\begin{lemma}
\label{lem:left-edge-lemma}
Let $T_1$ be the subtree of $T(G,avg(\by))$ that contains $y_1$ and assume that $Opt \notin T_1$. Let $(G',\bx')$ be a transformation of $(G,\bx)$ that, for every agent $i$ for which $y_i = y_1$, locates $x'_i$ at $y_1$. Then, the approximation ratio obtained by the mechanism for $(G',\bx')$ cannot be better than the ratio obtained for ($G,\bx)$.
\end{lemma}

\begin{proof}
Let $\delta=d(x_i,x'_i)$. The difference between the social cost obtained by the mechanism for $(G,\bx)$ and for $(G',\bx')$ is given by
\begin{align*}
&\sum_{j \in N} \frac{1}{2n} \left( 2d(x_i,y_j) \delta - \delta^2  \right) + \frac{1}{2} \left(  2d(x_i,avg(\by)) \delta - \delta^2 \right) = \\
&- \delta^2 + \sum_{j \in N} \frac{1}{2n} 2(d(x_i,y_1) + d(y_1,y_j) ) \delta + \frac{1}{2} 2( d(x_i,y_1) + d(y_1,avg(\by))) \delta = \\
&- \delta^2 + 2d(x_i,y_1)\delta + \sum_{j \in N} \frac{1}{2n} 2 d(y_1,y_j) \delta + \frac{1}{2} 2  d(y_1,avg(\by)) \delta =\\
&- \delta^2 + 2d(x_i,y_1)\delta + 2d(y_1,avg(\by))\delta,
\end{align*}

where the last equality follows from $\sum_{j \in N}d(y_1,y_j) = nd(y_1,avg(\by))$.

In contrast, it is easy to verify that the optimal cost decreases by at least $2d(x_i,y_1)\delta + 2d(y_1,Opt)\delta - \delta^2$.
As $d(y_1,avg(\by)) \leq d(y_1,Opt)$, the assertion follows.
\end{proof}


\begin{lemma}
\label{lem:right-edge-lemma}
Let $T_{Opt}$ be the subtree in $T(G,y_n)$ that contains $Opt$. Let $N_0 = \{i \in N: y_i = y_n, x_i \notin T_{Opt} \}$, and let $\sigma = \sum_{i \in N_0} d(x_i,y_n)$. Let $(G',\bx')$ be a transformation of $(G,\bx)$ that, for each $i \in N_0$, creates an edge of length $\frac{\sigma}{|N_0|}$, rooted at $y_n$, and locates $x'_i$ at its tip. Then, the approximation ratio obtained by the mechanism for $(G',\bx')$ cannot be better than the ratio obtained for $(G,\bx)$.
\end{lemma}

\begin{proof}
We first note that for each $i \in N_0$, $y'_i$ did not move. Therefore, the locations that have a positive probability of being chosen by the randomized mechanism do not change.

Let $\delta_i$ denote the difference in distances of agent $i$ from $y_n$, induced by the transition to $(G',\bx')$, i.e., $\delta_i = d(y_n,x'_i) - d(y_n,x_i)$. Recall that $d(y_n,x'_i) = \frac{\sigma}{|N_0|}$, thus it is easy to see that $\sum(\delta_i) = 0$.
It follows that the increase of cost in $(G',\bx')$ is as follows.
\begin{align*}
&\sum_{i \in N_0,\delta_i \geq 0} \left( \sum_{j \in N} \frac{1}{2n} \left( 2d(x_i,y_j) \delta_i + \delta_i^2  \right) + \frac{1}{2} \left(  2d(x_i,avg(\by)) \delta_i + \delta_i^2 \right) \right) &+ \\
&\sum_{i \in N_0,\delta_i < 0} \left( \sum_{j \in N} \frac{1}{2n} \left( 2d(x_i,y_j) \delta_i + \delta_i^2  \right) + \frac{1}{2} \left(  2d(x_i,avg(\by)) \delta_i + \delta_i^2 \right) \right) &= \\
&\sum_{i \in N_0,\delta_i \geq 0} \left( 2d(x_i,avg(\by)) \delta_i + \delta_i^2  \right) +
\sum_{i \in N_0,\delta_i < 0} \left(    2d(x_i,avg(\by)) \delta_i + \delta_i^2 \right) &= \\
&\sum_{i \in N_0,\delta_i \geq 0} \left( 2d(x'_i,avg(\by)) \delta_i - \delta_i^2  \right) +
\sum_{i \in N_0,\delta_i < 0} \left(     2d(x'_i,avg(\by)) \delta_i - \delta_i^2 \right) &= \\
&\sum_{i \in N_0} - \delta_i^2,
\end{align*}
Where the last two equalities are due to $d(x_i,avg(\by))+\delta = d(x'_i,avg(\by))$ and to $\sum(\delta_i) = 0$.
We note that in the transition from $(G,\bx)$ to $(G',\bx')$, $Opt$'s location does not change (see \lemref{lem:weighted-opt-derivative}). Therefore, in the same manner, the optimal cost decreases by $\sum_{i \in N} \delta_i^2$ as well. It follows directly that the approximation ratio cannot decrease in the transformation.
\end{proof}


\begin{lemma}
\label{lem:M-edge-lemma}
Assume that $\bar{Opt}$ is located at $y_n$, but $Opt$ is not, and let $T_{Opt}$ be the subtree in $T(G,\bar{Opt})$ that contains $Opt$.
Let $N_1 = \{i \in N : x_i \in T_{Opt} \}$, and let $\sigma = \sum_{i \in N_1} d(x_i,\bar{Opt})$. Let $(G',\bx')$ be a transformation of $(G,\bx)$ that creates an edge of length $\frac{\sigma}{|N_1|}$, rooted at $\bar{Opt}$, and locates all of the agents in $N_1$ at its tip. Then, the approximation ratio obtained by the mechanism for $(G',\bx')$ cannot be better than the ratio obtained for $(G,x)$.
\end{lemma}

\begin{proof}
We shall prove the lemma in two steps.
In particular, the instance $(G',\bx')$ shall be induced by $(G,\bx)$ through an intermediate instance $(G'',\bx'')$, as follows.
Let $j \in N_1$ denote the agent whose location is the farthest away from $\bar{Opt}$, i.e., $j = argmax_{i \in N_1} d(x_i,\bar{Opt})$, and for each $i \in N_1$, let $x''_i$ be located on $path(x_j,\bar{Opt})$, distanced $d(x_i,\bar{Opt})$ from $\bar{Opt}$. It is easy to verify that the cost of the randomized mechanism does not change in the transformation, while when examining the optimal location, clearly its induced cost decreases, as agents in $N_1$ could only move closer to it. It follows that the approximation ratio increased by the transformation. 
Next, we make the transformation from $(G'',\bx'')$ to $(G',\bx')$, and show that the approximation ratio keeps getting worse.

Using the notation of the proof of \lemref{lem:right-edge-lemma}, let $\delta_i$ denote the difference of distances of agent $i$ from $\bar{Opt}$, induced by the transformation from $(G'',\bx'')$ to $(G',\bx')$. Following the steps of the proof of \lemref{lem:right-edge-lemma}, it follows that the mechanism's induced cost is reduced by $\sum_{i \in N_1} \delta_i^2$. We next show that the optimal cost decreases by the same amount.

It is clear that $Opt''$ divides $G''$ to exactly two subtrees. Let $T_1,T_2 \in T(G'',Opt'')$ be the two subtrees induced by $Opt''$. By \lemref{lem:weighted-opt-derivative} it follows that $\sum_{i \in T_1} d(x''_i,Opt'') = \sum_{i \in T_2} d(x''_i,Opt'')$. After averaging, it is easy to verify that $\sum_{i \in T_1} d(x'_i,Opt'') = \sum_{i \in T_2} d(x'_i,Opt'')$, and thus by \lemref{lem:weighted-opt-derivative} it follows that the optimal location does not change. Therefore, the same calculation from \lemref{lem:right-edge-lemma} holds for the optimal location, and thus the optimal cost decreases by $\sum_{i \in N_1} \delta_i^2$ as well. The assertion of the lemma follows.
\end{proof}


\begin{corollary}
\label{cor:M-edge-open-lemma}
Let $p$ be some location on the open interval $path(y_1,y_n)$. Let $N_2 =\{ i \in N: y_i = p \}$, and let $\sigma = \sum_{i \in N_2} d(x_i,p)$. Let $(G',\bx')$ be a transformation of $(G,\bx)$ that creates an edge of length $\frac{\sigma}{|N_2|}$, rooted at $p$, and locates all of the agents from $N_2$ at its tip. Then, the approximation ratio obtained by the mechanism for $(G',\bx')$ cannot be better than the ratio obtained for $(G,x)$.
\end{corollary}

\begin{proof}
By observing that the process does not change the locations of $y_1,y_2,\ldots,y_n$, the proof is identical to the one of \lemref{lem:M-edge-lemma}.
\end{proof}


\begin{lemma}
\label{lem:pushing-M}
Assume that for each $i$ such that $y_i \neq y_n$, it holds that $y_i = x_i$. Let $N_3 = \{ i \in N : y_i = y_n, x_i \neq y_n \}$, and assume that for any $j,k \in N_3$ it holds that $d(x_j,y_n) = d(x_k,y_n)$ and that $x_j$ and $x_k$ are not located on the same subtree induced by $T(G,y_n)$. Then, the approximation ratio obtained by the mechanism is at most $1 \frac{1}{2}$.
\end{lemma}

\begin{proof}
Let $\gamma=d(avg(\by),y_n)$, and let $\delta$ denote the distance between agents in $N_3$ and $y_n$. Note that $Opt$ can be located only on $path(y_1,y_n)$. We consider the following two cases: the case in which $opt$ is located at $y_n$, and the case in which it is located somewhere else on $path(y_1,y_n)$. We start with the latter case.

Due to \lemref{lem:weighted-opt-derivative}, it follows that $\sum_{i \notin N_3} d(x_i,y_n) > \sum_{i \in N_3} d(x_i,y_n)$. Since
\begin{align*}
\sum_{i \notin N_3} d(x_i,y_n) = \sum_{i \notin N_3} d(y_i,y_n) = \sum_{i \in N} d(y_i,y_n) = |N|d(avg(\by),y_n),
\end{align*}
and $\sum_{i \in N_3} d(x_i,y_n) = |N_3|\delta$, it follows that
\begin{equation}
\label{eq:1000}
|N|\gamma > |N_3|\delta.
\end{equation}

Let $\bx'$ be derived from $\bx$, such that for any $i \in N_3$, $x'_i = y_n$. We note that all of the agents in $\bx'$ are located on $path(y_1,y_n)$, and that the mechanism described in \secref{subsec:SP-mechanisms-line-random} is identical to Mechanism \randdgm{} with $q=2/3$, when activated on $\bx'$. Thus, due to \thref{thm:random-1.5-approx} it follows that the approximation ratio is at most $1.5$, i.e, $\frac{sc(f(\bx'),\bx')}{sc(Opt',\bx')} \leq 1.5$. We note that for $\bx'$, $avg(\by)$ and $Opt$ are located at the same location.
When making the transformation to $\bx$, the mechanism's possible locations do not change, and it is easy to see that the mechanism's induced cost increases by $|N_3|(2 \gamma \delta + \delta^2)$. The change of the optimal cost can be calculated by checking how the original optimal location's induced cost increases by the transformation, and how it decreases by smoothly moving the location to the new optimal location, i.e., by calculating $(sc(Opt',\bx) - sc(Opt',\bx')) - (sc(Opt',\bx) - sc(Opt,\bx))$. Thus, using \lemref{lem:differnce-of-costs-two-on-line}, the optimal cost increases by
\begin{equation}
\label{eq:1001}
|N_3|(2 \gamma \delta + \delta^2) - \frac{|N_3|^2 \delta^2}{|N|}.
\end{equation}

Therefore, the approximation ratio of $G$ and $\bx$ is as follow.
\begin{equation}
\label{eq:1002}
\frac{sc(f(\bx'),\bx') + |N_3|(2 \gamma \delta + \delta^2)}{sc(Opt',\bx') + |N_3|(2 \gamma \delta + \delta^2) - \frac{|N_3|^2 \delta^2}{|N|}}.
\end{equation}

It follows that it is sufficient to prove that $\frac{|N_3|(2 \gamma \delta + \delta^2)}{|N_3|(2 \gamma \delta + \delta^2) - \frac{|N_3|^2 \delta^2}{|N|}} \leq 1.5$.  By a simple rearrangement, this is equivalent to proving that $3\frac{|N_3| \delta}{|N|} \leq 2 \gamma  + \delta$.
Using \eqref{eq:1000}, it follows that
\begin{align*}
3\frac{|N_3| \delta}{|N|} \leq 2\frac{|N_3| \delta}{|N|} + \delta \leq 2\gamma + \delta,
\end{align*}
and the assertion follows for the case in which $Opt$ is located on $path(y_1,y_n)$, excluding $y_n$.

We now turn to prove the lemma for the case in which $Opt$ is located at $y_n$. We follow the steps of the proof of the first case, only this time the optimal location moves from $avg(\by)$ to $y_n$, and thus in the transition, the optimal cost increases by $|N_3|(2 \gamma \delta + \delta^2) - |N| \left( 2\frac{|N_3|\delta}{|N|}\gamma - \gamma^2 \right)$.
Therefore, in the same manner, it is sufficient to prove that
\begin{equation}
\label{eq:1003}
\frac{|N_3|(2 \gamma \delta + \delta^2)}{|N_3|(2 \gamma \delta + \delta^2) - |N| \left( 2\frac{|N_3|\delta}{|N|}\gamma - \gamma^2 \right)} \leq 1.5,
\end{equation}
which is equivalent to
\begin{equation}
\label{eq:1004}
3|N| \left( 2\frac{|N_3|\delta}{|N|}\gamma - \gamma^2 \right) \leq |N_3|(2 \gamma \delta + \delta^2).
\end{equation}
Clearly, $|N_3| \leq \frac{2}{3}|N| \leq \frac{3}{4}|N|$. It follows that
\begin{equation}
\label{eq:1005}
3|N| \left( 2\frac{|N_3|\delta}{|N|}\gamma - \gamma^2 \right) = 6|N_3|\delta \gamma - 3|N| \gamma^2 \leq 6|N_3|\delta \gamma - 4|N_3| \gamma^2.
\end{equation}
Using \eqref{eq:1005} in \eqref{eq:1004}, it is sufficient to show that
\begin{align*}
6|N_3|\delta \gamma - 4|N_3| \gamma^2 \leq |N_3|(2 \gamma \delta + \delta^2),
\end{align*}
which is equivalent to showing that $(2\gamma - \delta)^2 \geq 0$. The assertion follows.
\end{proof}


\begin{lemma}
\label{lem:ratio-1.82}
Let $G$ and $\bx$ be a graph and a location profile as follows. $Q \subset N$ agents are located at $y_1$; $P \subset N$ agents are equally distanced from $y_n$, each on a different subtree of $T(G,y_n)$, such that for each agent $i \in P$ it holds that $y_i = y_n$; $M \subset N$ agents are located at the same location, distanced $\delta$ from $path(y_1,y_n)$, and the rest of the agents (i.e., agents in $N \backslash \{P \cup Q \cup M\}$) are scattered along $path(y_1,y_n)$.
Then, the approximation ratio obtained by the mechanism is at most $1.83$.
\end{lemma}

\begin{proof}
We start with a few preparations. First of all, we normalize $d(y_1,y_n)$ to be $1$, and rescale the graph accordingly. Given $i \in M$ let $y_M = y_i$. Let $\beta$ denote the distance between $y_1$ and $y_M$, and assume $d(y_1,y_M) \leq d(y_M,y_n)$ (the proof works analogously for $d(y_1,y_M) \geq d(y_M,y_n)$). Let $\bx'$ be defined as follows. for each $i \in M$, let $x'_i = y_i$; for the other agents, let $x'_i = x_i$.

Due to \lemref{lem:pushing-M}, it holds that
\begin{equation}
\label{eq:1006}
\frac{sc(f(\bx'),\bx')}{sc(Opt',\bx')} \leq 1\frac{1}{2}.
\end{equation}
Let $\Delta_{Opt}=sc(Opt,\bx) - sc(Opt',\bx')$ and $\Delta_{Ran}=sc(f(\bx),\bx) - sc(f(\bx'),\bx')$.

We start with providing an upper bound for $\Delta_{ran}$. 

We first observe that $|Q|$ should contain the minimal number of agents in order to maximize the costs' difference. Indeed, each agent that could be relocated at $P$, keeping $y_1$ fixed (meaning, keeping $|Q| \geq |N|/3$), would increase the costs' difference. The same holds for any agent that is located anywhere on the open $path(y_1,y_n)$.

We next observe that the distance between $y_M$ and $avg(\by)$ cannot be higher than $\frac{2}{3} - \frac{|M|}{|N|}$, due to the following. Assuming that $avg(\by)$ is on $path(y_M,y_n)$ (the proof works identically for $avg(\by)$ on $path(y_1,y_M)$), let $\bx^0$ be constructed such that $\bx^0$ differs from $\bx$ such that for each $i \in M$, $x^0_i = y_n$. Its easy to verify that $d(avg(\by^0),y_M) - \frac{|M|}{|N|} d(y_n,y_M) = d(avg(\by),y_M)$, as in the transition of each agent in $|M|$ from $\bx^0$ to $\bx$, the weighted optimal location transitions by $\frac{1}{|N|} d(y_n,y_M)$. Therefore, the distance between $y_M$ and $avg(\by)$ is bounded as follows:
\begin{equation}
\label{eq:1007}
d(y_M,avg(\by)) = d(avg(\by^0),y_M) - \frac{|M|}{|N|} d(y_n,y_M) = d(y_n,y_M) \left( 1-\frac{|M|}{|N|} \right) - d(y_n,avg(\by^0)) \leq \frac{2}{3}-\frac{|M|}{|N|},
\end{equation}
where the last equality is due to the fact that because there are at least $\frac{|N|}{3}$ agents that are projected on each side of $path(y_1,y_n)$, $d(y_n,avg(\by^0))$ must be bigger than $\frac{1}{3}$.

We can now bound the difference of costs between $sc(f(\bx),\bx)$ and $sc(f(\bx'),\bx')$, by the following expression (an explanation follows):
\begin{equation}
\label{eq:1008}
|M| \left[ \frac{|M|}{2|N|} \delta^2 + \frac{1}{6}(2 \beta \delta + \delta^2 ) + \left( \frac{1}{3} - \frac{|M|}{2|N|} \right)(2(1- \beta)\delta + \delta^2) + \frac{1}{2}(2 \left(\frac{2}{3}-\frac{|M|}{|N|} \right) \delta + \delta^2)  \right].
\end{equation}
Each summand represents the expected difference of costs induced by each location with a positive probability. The first summand stands for agents in $M$ (the difference of costs when choosing agent in $M$ is $\delta^2$, and the probability is $\frac{|M|}{2|N|}$); the second summand stands for agents in $Q$, and as shown earlier, in order to bound the difference, we assume that $|Q| = \frac{1}{3}$; the next summand stands for $P$, and the probability of choosing an agent in $P$ is calculated given $|M|$; and $avg(\by)$ is the last summand. Naturally, the whole expression is multiplied by $|M|$, as the expected difference of costs is multiplied by the number of agents that make the transition.

By rearranging expression \eqref{eq:1008}, we get
\begin{equation}
\label{eq:1009}
|M| \left[ \delta^2 + \beta \left( \frac{|M| \delta}{|N|} - \frac{\delta}{3} \right) + \frac{4}{3}\delta - \frac{2 |M| \delta}{|N|}   \right].
\end{equation}

Since $\frac{|M|}{|N|} \leq \frac{1}{3}$, it holds that $\beta \left( \frac{|M| \delta}{|N|} - \frac{\delta}{3} \right) \leq 0$. It, therefore, follows that
\begin{equation}
\label{eq:1010}
\Delta_{Ran} \leq |M| \left[ \delta^2 + \frac{4}{3}\delta - \frac{2 |M| \delta}{|N|}   \right].
\end{equation}

Therefore, we conclude the following.
\begin{equation}
\label{eq:1011}
\frac{sc(f(\bx),\bx)}{sc(Opt,\bx)} \leq \frac{sc(f(\bx'),\bx') + |M| \left[ \delta^2 + \frac{4}{3}\delta - \frac{2 |M| \delta}{|N|} \right]}{sc(Opt',\bx') + \Delta_{Opt}}.
\end{equation}

We next distinguish between two cases: one where $Opt$ is on $path(y_1,y_n)$, and the second in which it is not.
We start with the former case.
In this case, it is easy to verify that $\Delta_{Opt} \geq |M|\delta^2$. In addition, due to Equation \eqref{eq:1006}, it holds that $\frac{sc(f(\bx'),\bx') + 1\frac{1}{2}|M|\delta^2}{sc(Opt',\bx')+|M|\delta^2} \leq 1\frac{1}{2}$. Together with \eqref{eq:1011}, it follows that
\begin{equation}
\label{eq:1012}
\frac{sc(f(\bx),\bx)}{sc(Opt,\bx)} \leq 1\frac{1}{2}  + \frac{ |M| \left[ - \frac{\delta^2}{2} + \frac{4}{3}\delta - \frac{2 |M| \delta}{|N|} \right]}{sc(Opt',\bx') + |M|\delta^2}.
\end{equation}
Next, we point out that $sc(Opt',\bx') \geq \frac{|N|}{6}$. This is due to the fact that each of $Q$ and $P$ must contain at least $|N|/3$ agents each, and thus even if all the remaining agents would be located on $Opt'$, and $P$ and $Q$ would be located as close to $Opt$ as possible, the cost could not be lower than $\frac{|N|}{6}$. Therefore, we conclude that
\begin{equation}
\label{eq:1013}
\frac{sc(f(\bx),\bx)}{sc(Opt,\bx)} \leq 1\frac{1}{2}  + \frac{ |M| \left[ - \frac{\delta^2}{2} + \frac{4}{3}\delta - \frac{2 |M| \delta}{|N|} \right]}{\frac{|N|}{6} + |M|\delta^2}.
\end{equation}
A numeric analysis would find that the maximal value obtained for the right-hand side of \eqref{eq:1013} is $1.82085\dots$ (obtained by fixing $\delta = 0.54144\dots$ and $|M| = \frac{2}{9}|N|$).

It is left to analyze the latter case, namely the case in which $Opt$ is not on $path(y_1,y_n)$. Due to \corref{cor:opt-flattening}, when moving the facility location from $\bar{Opt}$ to $Opt$, its induced cost is reduced by $|N|d(\bar{Opt},Opt)^2$. Let ${T_1,T_2} \in T(G,Opt)$, such that the agents from $M$ are in $T_1$. Due to \lemref{lem:weighted-opt-derivative}, it holds that $\sum_{x_i \in T_1}d(x_i,Opt) = \sum_{x_i \in T_2}d(x_i,Opt)$. By substituting the left-hand side and the right-hand side, we get 
\begin{equation}
\label{eq:1130}
|M|(\delta - d(\bar{Opt},Opt)) = (1-|M|)d(\bar{Opt},Opt) +  \sum_{x_i \in T_2}d(x_i,\bar{Opt}).
\end{equation}

Since $\sum_{x_i \in T_2}d(x_i,\bar{Opt}) \geq \frac{|N|}{3}$, and by rearranging we get that
\begin{equation}
\label{eq:1131}
\frac{|M|}{|N|}\delta - \frac{1}{3} \geq d(\bar{Opt},Opt).
\end{equation}

Now we can bound $\Delta_{Opt}$ as follows. 
\begin{align*}
\Delta_{Opt} &= [sc(\bar{Opt},\bx) - sc(Opt',\bx')] - [sc(\bar{Opt},\bx) - sc(Opt,\bx)] \\
&\geq [sc(\bar{Opt},\bx) - sc(\bar{Opt},\bx')] - [sc(\bar{Opt},\bx) - sc(Opt,\bx)] \geq |M|\delta^2 - |N|\left( \frac{|M|}{|N|} \delta - \frac{1}{3} \right) ^2.
\end{align*}
Following the steps of the former case, it follows that
\begin{equation}
\label{eq:1015}
\frac{sc(f(\bx),\bx)}{sc(Opt,\bx)} \leq 1\frac{1}{2}  + \frac{ |M| \left[ - \frac{\delta^2}{2} + \frac{4}{3}\delta - \frac{2 |M| \delta}{|N|} \right] + \frac{3}{2}|N|\left( \frac{|M|}{|N|} \delta - \frac{1}{3} \right) ^2}{\frac{|N|}{6} + |M|\delta^2 - |N|\left( \frac{|M|}{|N|} \delta - \frac{1}{3} \right) ^2}.
\end{equation}

A numeric analysis would find that the maximal value obtained for the right-hand side of \eqref{eq:1015} is $1.611\dots$ (obtained by fixing $\delta = 1.0446\dots$ and $|M| = 0.32|N|$). The assertion of the lemma follows.
\end{proof}

\parindent 10mm

\end{document}

%% file: draw-tree-median-before.tex
\begin{tikzpicture}[scale=1]

\tikzstyle{dot}=[circle,draw=black,fill=white,thin,inner sep=0pt,minimum size=3mm]
\tikzstyle{point}=[circle,draw=black,fill=white,thin,inner sep=0pt,minimum size=0mm]
\tikzstyle{blackdot}=[circle,draw=black,fill=black,thin,
inner sep=0pt,minimum size=1.5mm]


\node (mu) at (0,0) [dot] {};
\node (mu) at (0,0) [blackdot] {};
\node at (0.2,-0.4) {\small{$\mu$}};

\node (opt) at (2,0) [blackdot] {};
\node at (1.8,-0.4) {\small{$opt$}};

\path (mu) edge node [above] {} (opt);

\node (i1) at (0.7,0) [point] {};
\node (i2) at (1,0) [point] {};
\path (i2) edge node [above] {$a$} (opt);

\node (xk) at (0.8,-1.4) [dot] {};
\node (xj) at (1,1.5) [dot] {\small{$x_j$}};

\path (i1) edge node [above] {} (xk);
\path (i2) edge node [right] {b} (xj);

\node (l0) at (-0.2,1.5) [dot] {};
\node (l1) at (-1.4,1.4) [dot] {};
\node (l2) at (-1,0) [dot] {};
\node (l3) at (-0.7,-0.7) [dot] {};
\node (l4) at (-0.2,-1.5) [dot] {};
\path (mu) edge node [above] {} (l0);
\path (mu) edge node [above] {} (l1);
\path (mu) edge node [above] {} (l2);
\path (mu) edge node [above] {} (l3);
\path (mu) edge node [above] {} (l4);

\node (r1) at (3.4,1.4) [dot] {};
\node (r2) at (4,0) [dot] {};
\node (r3) at (3.4,-1.4) [dot] {};
\path (opt) edge node [above] {} (r1);
\path (opt) edge node [above] {} (r2);
\path (opt) edge node [above] {} (r3);

\end{tikzpicture}

%% file: draw-tree-median-after.tex
\begin{tikzpicture}[scale=1]
\tikzstyle{dot}=[circle,draw=black,fill=white,thin,inner sep=0pt,minimum size=3mm]
\tikzstyle{point}=[circle,draw=black,fill=white,thin,inner sep=0pt,minimum size=0mm]
\tikzstyle{blackdot}=[circle,draw=black,fill=black,thin,
inner sep=0pt,minimum size=1.5mm]

\node (mu) at (0,0) [dot] {};
\node (mu) at (0,0) [blackdot] {};
\node at (0.2,-0.4) {\small{$\mu$}};

\node (opt) at (2,0) [blackdot] {};
\node at (1.8,-0.4) {\small{$opt$}};

\path (mu) edge node [above] {} (opt);

\node (i1) at (0.7,0) [point] {};

\node (xk) at (0.8,-1.4) [dot] {};
\node (xj) at (1.6,2) [dot] {\small{$x_j$}};

\path (i1) edge node [above] {} (xk);
\path (opt) edge [dashed] node [right] {a+b} (xj);

\node (l0) at (-0.2,1.5) [dot] {};
\node (l1) at (-1.4,1.4) [dot] {};
\node (l2) at (-1,0) [dot] {};
\node (l3) at (-0.7,-0.7) [dot] {};
\node (l4) at (-0.2,-1.5) [dot] {};
\path (mu) edge node [above] {} (l0);
\path (mu) edge node [above] {} (l1);
\path (mu) edge node [above] {} (l2);
\path (mu) edge node [above] {} (l3);
\path (mu) edge node [above] {} (l4);

\node (r1) at (3.4,1.4) [dot] {};
\node (r2) at (4,0) [dot] {};
\node (r3) at (3.4,-1.4) [dot] {};
\path (opt) edge node [above] {} (r1);
\path (opt) edge node [above] {} (r2);
\path (opt) edge node [above] {} (r3);

\end{tikzpicture}

%% file: draw-original-state.tex
\begin{tikzpicture}[scale=1]
\tikzstyle{dot}=[circle,draw=black,fill=white,thin,inner sep=0pt,minimum size=3mm]
\tikzstyle{point}=[circle,draw=black,fill=white,thin,inner sep=0pt,minimum size=0mm]
\tikzstyle{blackdot}=[circle,draw=black,fill=black,thin,inner sep=0pt,minimum size=1.5mm]

\node (y1) at (0,0) [blackdot] {};
\node at (0.2,-0.4) {\small{$y_1$}};

\node (yn) at (2,0) [blackdot] {};
\node at (2,0.4) {\small{$y_n$}};

\path (y1) edge node [above] {} (yn);

\node (ya) at (1.3,0) [point] {};
\node (xa) at (1.1,-0.5) [dot] {};
\path (xa) edge node [above] {} (ya);

\node (yi) at (0.5,0) [blackdot] {};
\node at (0.6,-0.4) {\small{$y_i$}};
\node (xi) at (0.5,1.2) [dot] {\small{$x_i$}};
\path (yi) edge node [above] {} (xi);


\node (optbar) at (1.5,0) [blackdot] {};
\node at (1.5,-0.4) {\small{$\bar{Opt}$}};
\node (agent2) at (1.5,2) [dot] {};
\node (opt) at (1.5,1) [blackdot] {};
\node at (1.9,1) {\small{$Opt$}};
\path (optbar) edge node [above] {} (agent2);

\node (agent1) at (1.2,1.6) [dot] {};
\node (agent3) at (1.8,1.6) [dot] {};
\node (split) at (1.5,1.5) [point] {};
\path (split) edge node [above] {} (agent1);
\path (split) edge node [above] {} (agent3);


\node (l0) at (-0.2,1.5) [dot] {};
\node (l1) at (-1.2,1.2) [dot] {};
\node (l2) at (-1,-0.3) [dot] {};
\node (l3) at (-0.7,-0.7) [dot] {};
\node (l4) at (-0.2,-1.0) [dot] {};
\path (y1) edge node [above] {} (l0);
\path (y1) edge node [above] {} (l1);
\path (y1) edge node [above] {} (l2);
\path (y1) edge node [above] {} (l3);
\path (y1) edge node [above] {} (l4);

\node (r1) at (2.6,1.2) [dot] {};
\node (r2) at (2.9,0) [dot] {};
\node (r3) at (2.3,-0.5) [dot] {};
\path (yn) edge node [above] {} (r1);
\path (yn) edge node [above] {} (r2);
\path (yn) edge node [above] {} (r3);

\draw[line width=5pt,->] (3.5,0) -- (4.5,0);

\end{tikzpicture}

%% file: draw-lemma1.tex
\begin{tikzpicture}[scale=1]
\tikzstyle{dot}=[circle,draw=black,fill=white,thin,inner sep=0pt,minimum size=3mm]
\tikzstyle{point}=[circle,draw=black,fill=white,thin,inner sep=0pt,minimum size=0mm]
\tikzstyle{blackdot}=[circle,draw=black,fill=black,thin,inner sep=0pt,minimum size=1.5mm]

\node (y1) at (0,0) [blackdot] {};
\node at (0.2,-0.4) {\small{$y_1$}};

\node (yn) at (2,0) [blackdot] {};
\node at (2,0.4) {\small{$y_n$}};

\path (y1) edge node [above] {} (yn);

\node (xa) at (0.9,0) [dot] {};

\node (xi) at (0.8,2.08) [dot] {\small{$x_i$}};
\path (optbar) edge node [above] {} (xi);


\node (optbar) at (1.5,0) [blackdot] {};
\node at (1.5,-0.4) {\small{$\bar{Opt}$}};
\node (agent2) at (1.5,2) [dot] {};
\node (opt) at (1.5,1) [blackdot] {};
\node at (1.9,1) {\small{$Opt$}};
\path (optbar) edge node [above] {} (agent2);

\node (agent1) at (1.2,1.6) [dot] {};
\node (agent3) at (1.8,1.6) [dot] {};
\node (split) at (1.5,1.5) [point] {};
\path (split) edge node [above] {} (agent1);
\path (split) edge node [above] {} (agent3);


\node (l0) at (-0.2,1.5) [dot] {};
\node (l1) at (-1.2,1.2) [dot] {};
\node (l2) at (-1,-0.3) [dot] {};
\node (l3) at (-0.7,-0.7) [dot] {};
\node (l4) at (-0.2,-1.0) [dot] {};
\path (y1) edge node [above] {} (l0);
\path (y1) edge node [above] {} (l1);
\path (y1) edge node [above] {} (l2);
\path (y1) edge node [above] {} (l3);
\path (y1) edge node [above] {} (l4);

\node (r1) at (2.6,1.2) [dot] {};
\node (r2) at (2.9,0) [dot] {};
\node (r3) at (2.3,-0.5) [dot] {};
\path (yn) edge node [above] {} (r1);
\path (yn) edge node [above] {} (r2);
\path (yn) edge node [above] {} (r3);

\end{tikzpicture}

%% file: draw-lemma1-arrow.tex
\begin{tikzpicture}[scale=1]
\tikzstyle{dot}=[circle,draw=black,fill=white,thin,inner sep=0pt,minimum size=3mm]
\tikzstyle{point}=[circle,draw=black,fill=white,thin,inner sep=0pt,minimum size=0mm]
\tikzstyle{blackdot}=[circle,draw=black,fill=black,thin,inner sep=0pt,minimum size=1.5mm]

\node (y1) at (0,0) [blackdot] {};
\node at (0.2,-0.4) {\small{$y_1$}};

\node (yn) at (2,0) [blackdot] {};
\node at (2,0.4) {\small{$y_n$}};

\path (y1) edge node [above] {} (yn);

\node (xa) at (0.9,0) [dot] {};

\node (xi) at (0.8,2.08) [dot] {\small{$x_i$}};
\path (optbar) edge node [above] {} (xi);


\node (optbar) at (1.5,0) [blackdot] {};
\node at (1.5,-0.4) {\small{$\bar{Opt}$}};
\node (agent2) at (1.5,2) [dot] {};
\node (opt) at (1.5,1) [blackdot] {};
\node at (1.9,1) {\small{$Opt$}};
\path (optbar) edge node [above] {} (agent2);

\node (agent1) at (1.2,1.6) [dot] {};
\node (agent3) at (1.8,1.6) [dot] {};
\node (split) at (1.5,1.5) [point] {};
\path (split) edge node [above] {} (agent1);
\path (split) edge node [above] {} (agent3);


\node (l0) at (-0.2,1.5) [dot] {};
\node (l1) at (-1.2,1.2) [dot] {};
\node (l2) at (-1,-0.3) [dot] {};
\node (l3) at (-0.7,-0.7) [dot] {};
\node (l4) at (-0.2,-1.0) [dot] {};
\path (y1) edge node [above] {} (l0);
\path (y1) edge node [above] {} (l1);
\path (y1) edge node [above] {} (l2);
\path (y1) edge node [above] {} (l3);
\path (y1) edge node [above] {} (l4);

\node (r1) at (2.6,1.2) [dot] {};
\node (r2) at (2.9,0) [dot] {};
\node (r3) at (2.3,-0.5) [dot] {};
\path (yn) edge node [above] {} (r1);
\path (yn) edge node [above] {} (r2);
\path (yn) edge node [above] {} (r3);

\draw[line width=5pt,->] (3.5,0) -- (4.5,0);

\end{tikzpicture}

%% file: draw-lemma2.tex
\begin{tikzpicture}[scale=1]
\tikzstyle{dot}=[circle,draw=black,fill=white,thin,inner sep=0pt,minimum size=3mm]
\tikzstyle{point}=[circle,draw=black,fill=white,thin,inner sep=0pt,minimum size=0mm]
\tikzstyle{blackdot}=[circle,draw=black,fill=black,thin,inner sep=0pt,minimum size=1.5mm]

\node (y5) at (-0.2,0.2) [dot] {};
\node (y4) at (-0.15,0.15) [dot] {};
\node (y3) at (-0.1,0.1) [dot] {};
\node (y2) at (-0.05,0.05) [dot] {};
\node (y1) at (0,0) [dot] {};

\node at (0.2,-0.4) {\small{$y_1$}};

\node (yn) at (2,0) [blackdot] {};
\node at (2,0.4) {\small{$y_n$}};

\path (y1) edge node [above] {} (yn);

\node (xa) at (0.9,0) [dot] {};

\node (xi) at (0.8,2.08) [dot] {\small{$x_i$}};
\path (optbar) edge node [above] {} (xi);


\node (optbar) at (1.5,0) [blackdot] {};
\node at (1.5,-0.4) {\small{$\bar{Opt}$}};
\node (agent2) at (1.5,2) [dot] {};
\node (opt) at (1.5,1) [blackdot] {};
\node at (1.9,1) {\small{$Opt$}};
\path (optbar) edge node [above] {} (agent2);

\node (agent1) at (1.2,1.6) [dot] {};
\node (agent3) at (1.8,1.6) [dot] {};
\node (split) at (1.5,1.5) [point] {};
\path (split) edge node [above] {} (agent1);
\path (split) edge node [above] {} (agent3);


\node (r1) at (2.6,1.2) [dot] {};
\node (r2) at (2.9,0) [dot] {};
\node (r3) at (2.3,-0.5) [dot] {};
\path (yn) edge node [above] {} (r1);
\path (yn) edge node [above] {} (r2);
\path (yn) edge node [above] {} (r3);

\end{tikzpicture}

%% file: draw-lemma2-arrow.tex
\begin{tikzpicture}[scale=1]
\tikzstyle{dot}=[circle,draw=black,fill=white,thin,inner sep=0pt,minimum size=3mm]
\tikzstyle{point}=[circle,draw=black,fill=white,thin,inner sep=0pt,minimum size=0mm]
\tikzstyle{blackdot}=[circle,draw=black,fill=black,thin,inner sep=0pt,minimum size=1.5mm]

\node (y5) at (-0.2,0.2) [dot] {};
\node (y4) at (-0.15,0.15) [dot] {};
\node (y3) at (-0.1,0.1) [dot] {};
\node (y2) at (-0.05,0.05) [dot] {};
\node (y1) at (0,0) [dot] {};

\node at (0.2,-0.4) {\small{$y_1$}};

\node (yn) at (2,0) [blackdot] {};
\node at (2,0.4) {\small{$y_n$}};

\path (y1) edge node [above] {} (yn);

\node (xa) at (0.9,0) [dot] {};

\node (xi) at (0.8,2.08) [dot] {\small{$x_i$}};
\path (optbar) edge node [above] {} (xi);


\node (optbar) at (1.5,0) [blackdot] {};
\node at (1.5,-0.4) {\small{$\bar{Opt}$}};
\node (agent2) at (1.5,2) [dot] {};
\node (opt) at (1.5,1) [blackdot] {};
\node at (1.9,1) {\small{$Opt$}};
\path (optbar) edge node [above] {} (agent2);

\node (agent1) at (1.2,1.6) [dot] {};
\node (agent3) at (1.8,1.6) [dot] {};
\node (split) at (1.5,1.5) [point] {};
\path (split) edge node [above] {} (agent1);
\path (split) edge node [above] {} (agent3);


\node (r1) at (2.6,1.2) [dot] {};
\node (r2) at (2.9,0) [dot] {};
\node (r3) at (2.3,-0.5) [dot] {};
\path (yn) edge node [above] {} (r1);
\path (yn) edge node [above] {} (r2);
\path (yn) edge node [above] {} (r3);

\draw[line width=5pt,->] (3.5,0) -- (4.5,0);

\end{tikzpicture}

%% file: draw-lemma3.tex
\begin{tikzpicture}[scale=1]
\tikzstyle{dot}=[circle,draw=black,fill=white,thin,inner sep=0pt,minimum size=3mm]
\tikzstyle{point}=[circle,draw=black,fill=white,thin,inner sep=0pt,minimum size=0mm]
\tikzstyle{blackdot}=[circle,draw=black,fill=black,thin,inner sep=0pt,minimum size=1.5mm]

\node (y5) at (-0.2,0.2) [dot] {};
\node (y4) at (-0.15,0.15) [dot] {};
\node (y3) at (-0.1,0.1) [dot] {};
\node (y2) at (-0.05,0.05) [dot] {};
\node (y1) at (0,0) [dot] {};

\node at (0.2,-0.4) {\small{$y_1$}};

\node (yn) at (2,0) [blackdot] {};
\node at (2,0.4) {\small{$y_n$}};

\path (y1) edge node [above] {} (yn);

\node (xa) at (0.9,0) [dot] {};

\node (xi) at (0.8,2.08) [dot] {\small{$x_i$}};
\path (optbar) edge node [above] {} (xi);


\node (optbar) at (1.5,0) [blackdot] {};
\node at (1.5,-0.4) {\small{$\bar{Opt}$}};
\node (agent2) at (1.5,2) [dot] {};
\node (opt) at (1.5,1) [blackdot] {};
\node at (1.9,1) {\small{$Opt$}};
\path (optbar) edge node [above] {} (agent2);

\node (agent1) at (1.2,1.6) [dot] {};
\node (agent3) at (1.8,1.6) [dot] {};
\node (split) at (1.5,1.5) [point] {};
\path (split) edge node [above] {} (agent1);
\path (split) edge node [above] {} (agent3);


\node (r1) at (2.5,0.8) [dot] {};
\node (r2) at (2.9,0) [dot] {};
\node (r3) at (2.5,-0.7) [dot] {};
\path (yn) edge node [above] {} (r1);
\path (yn) edge node [above] {} (r2);
\path (yn) edge node [above] {} (r3);

\end{tikzpicture}

%% file: draw-lemma3-arrow.tex
\begin{tikzpicture}[scale=1]
\tikzstyle{dot}=[circle,draw=black,fill=white,thin,inner sep=0pt,minimum size=3mm]
\tikzstyle{point}=[circle,draw=black,fill=white,thin,inner sep=0pt,minimum size=0mm]
\tikzstyle{blackdot}=[circle,draw=black,fill=black,thin,inner sep=0pt,minimum size=1.5mm]

\node (y5) at (-0.2,0.2) [dot] {};
\node (y4) at (-0.15,0.15) [dot] {};
\node (y3) at (-0.1,0.1) [dot] {};
\node (y2) at (-0.05,0.05) [dot] {};
\node (y1) at (0,0) [dot] {};

\node at (0.2,-0.4) {\small{$y_1$}};

\node (yn) at (2,0) [blackdot] {};
\node at (2,0.4) {\small{$y_n$}};

\path (y1) edge node [above] {} (yn);

\node (xa) at (0.9,0) [dot] {};

\node (xi) at (0.8,2.08) [dot] {\small{$x_i$}};
\path (optbar) edge node [above] {} (xi);


\node (optbar) at (1.5,0) [blackdot] {};
\node at (1.5,-0.4) {\small{$\bar{Opt}$}};
\node (agent2) at (1.5,2) [dot] {};
\node (opt) at (1.5,1) [blackdot] {};
\node at (1.9,1) {\small{$Opt$}};
\path (optbar) edge node [above] {} (agent2);

\node (agent1) at (1.2,1.6) [dot] {};
\node (agent3) at (1.8,1.6) [dot] {};
\node (split) at (1.5,1.5) [point] {};
\path (split) edge node [above] {} (agent1);
\path (split) edge node [above] {} (agent3);


\node (r1) at (2.5,0.8) [dot] {};
\node (r2) at (2.9,0) [dot] {};
\node (r3) at (2.5,-0.7) [dot] {};
\path (yn) edge node [above] {} (r1);
\path (yn) edge node [above] {} (r2);
\path (yn) edge node [above] {} (r3);

\draw[line width=5pt,->] (3.5,0) -- (4.5,0);

\end{tikzpicture}

%% file: draw-lemma4.tex
\begin{tikzpicture}[scale=1]
\tikzstyle{dot}=[circle,draw=black,fill=white,thin,inner sep=0pt,minimum size=3mm]
\tikzstyle{point}=[circle,draw=black,fill=white,thin,inner sep=0pt,minimum size=0mm]
\tikzstyle{blackdot}=[circle,draw=black,fill=black,thin,inner sep=0pt,minimum size=1.5mm]

\node (y5) at (-0.2,0.2) [dot] {};
\node (y4) at (-0.15,0.15) [dot] {};
\node (y3) at (-0.1,0.1) [dot] {};
\node (y2) at (-0.05,0.05) [dot] {};
\node (y1) at (0,0) [dot] {};

\node at (0.2,-0.4) {\small{$y_1$}};

\node (yn) at (2,0) [blackdot] {};
\node at (2,0.4) {\small{$y_n$}};

\path (y1) edge node [above] {} (yn);

\node (xa) at (0.9,0) [dot] {};


\node (optbar) at (1.5,0) [blackdot] {};
\node at (1.5,-0.4) {\small{$\bar{Opt}$}};

\node (agent1) at (1.35,1.95) [dot] {\small{$x_i$}};
\node (agent3) at (1.4,1.9) [dot] {};
\node (agent4) at (1.45,1.85) [dot] {};
\node (agent2) at (1.5,1.8) [dot] {};

\node (opt) at (1.5,1) [blackdot] {};
\node at (1.9,1) {\small{$Opt$}};
\path (optbar) edge node [above] {} (agent2);


\node (r1) at (2.5,0.8) [dot] {};
\node (r2) at (2.9,0) [dot] {};
\node (r3) at (2.5,-0.7) [dot] {};
\path (yn) edge node [above] {} (r1);
\path (yn) edge node [above] {} (r2);
\path (yn) edge node [above] {} (r3);

\end{tikzpicture}